\documentclass[twocolumn,10pt]{IEEEtran}

\usepackage[acronym]{glossaries}
\usepackage{amssymb}
\usepackage{amsmath}
\usepackage{amsthm}
\usepackage{graphicx}
\usepackage{comment}
\newacronym{awgn}{AWGN}{additive white Gaussian noise}
\newacronym{siso}{SISO}{single-input single-output}
\newacronym{mimo}{MIMO}{multiple-input multiple-output}
\newacronym{csi}{CSI}{channel state information}
\newacronym{upa}{UPA}{uniform planar array}
\newacronym{ula}{ULA}{uniform linear array}
\newacronym{ofdm}{OFDM}{orthogonal frequency division multiplexing}
\newacronym{wpt}{WPT}{wireless power transfer}
\newacronym{swipt}{SWIPT}{simultaneous wireless information and power transfer}
\newacronym{mec}{MEC}{mobile edge computing}
\newacronym{rsma}{RSMA}{rate splitting multiple access}
\newacronym{dfrc}{DFRC}{dual-function radar-communication}
\newacronym{iid}{i.i.d.}{independent and identically distributed}
\newacronym{em}{EM}{electromagnetic}
\newacronym{los}{LoS}{line-of-sight}
\newacronym{bs}{BS}{base station}
\newacronym{ue}{UE}{user equipment}

\newacronym{mrt}{MRT}{maximum ratio transmission}
\newacronym{rf}{RF}{radio frequency}

\usepackage{xcolor}
\usepackage{subfigure}

\newtheorem{proposition}{Proposition}

\begin{document}

\title{Localized and Distributed Beyond Diagonal Reconfigurable Intelligent Surfaces with Lossy Interconnections: Modeling and Optimization}

\author{Matteo~Nerini,~\IEEEmembership{Member,~IEEE},
        Golsa~Ghiaasi,~\IEEEmembership{Member,~IEEE},
        Bruno~Clerckx,~\IEEEmembership{Fellow,~IEEE}

\thanks{This work has been partially supported by UKRI grant EP/Y004086/1, EP/X040569/1, EP/Y037197/1, EP/X04047X/1, EP/Y037243/1.}
\thanks{M. Nerini and B. Clerckx are with the Department of Electrical and Electronic Engineering, Imperial College London, London SW7 2AZ, U.K. (e-mail: \{m.nerini20, b.clerckx\}@imperial.ac.uk).}
\thanks{G. Ghiaasi is with Silicon Austria Labs (SAL), Graz A-8010, Austria (e-mail: golsa.ghiaasi@ieee.org).}}

\maketitle

\begin{abstract}
Reconfigurable intelligent surface (RIS) is a key technology to control the communication environment in future wireless networks.
Recently, beyond diagonal RIS (BD-RIS) emerged as a generalization of RIS achieving larger coverage through additional tunable impedance components interconnecting the RIS elements.
However, conventional RIS and BD-RIS can effectively serve only users in their proximity, resulting in limited coverage.
To overcome this limitation, in this paper, we investigate distributed RIS, whose elements are distributed over a wide region, in opposition to localized RIS commonly considered in the literature.
The scaling laws of distributed BD-RIS reveal that it offers significant gains over distributed conventional RIS and localized BD-RIS, enabled by its interconnections allowing signal propagation within the BD-RIS.
To assess the practical performance of distributed BD-RIS, we model and optimize BD-RIS with lossy interconnections through transmission line theory.
Our model accounts for phase changes and losses over the BD-RIS interconnections arising when the interconnection lengths are not much smaller than the wavelength.
Numerical results show that the performance of localized BD-RIS is only slightly impacted by losses, given the short interconnection lengths.
Besides, distributed BD-RIS can achieve orders of magnitude of gains over conventional RIS, even in the presence of low losses.
\end{abstract}

\glsresetall


\section{Introduction}

Reconfigurable intelligent surface (RIS) is a promising technology allowing the control of the propagation environment in wireless networks \cite{wu20a,wu21}.
RIS is a technology that involves deploying arrays consisting of numerous individually controllable elements.
The reflection coefficients of these elements can be adjusted in real-time to manipulate electromagnetic waves, allowing for dynamic control over the wireless communication environment and the realization of smart radio environments \cite{dir20}.
RIS presents significant appeal owing to its capacity to optimize signal propagation in a cost-effective manner and with minimal power consumption.

RIS has been used to improve wireless communication systems under several aspects, such as to minimize the transmit power \cite{wu19}, to maximize the achievable rate \cite{yan20}, to maximize the weighted sum-rate \cite{guo20a}, in narrow-band as well as wide-band systems \cite{li21}.
RIS has been also applied to further improve the data rate by modulating and encoding data into the reconfigurable elements \cite{bas20,guo20b}.
In addition to communication systems, RISs have been also employed to enhance the efficiency of \gls{wpt} systems \cite{fen22}, \gls{swipt} systems \cite{zha22}, and localization systems \cite{elz21}.
The practical problem of channel estimation in RIS-aided systems has been tackled in \cite{you20}.
Besides, the performance of practical RIS based on discrete reflection coefficients has been analyzed in \cite{wu20b}.
Finally, prototypes of RIS have been presented in \cite{dai20,zha21,rao22}.

Conventionally, RIS has been realized by reconfiguring each RIS element through a tunable load, leading to a diagonal phase shift matrix with limited flexibility.
To improve the limited flexibility of conventional RIS, beyond diagonal RIS (BD-RIS) has been recently proposed \cite{li23-1}.
The key novelty of BD-RIS is the presence of tunable impedance components interconnecting the RIS elements to each other.
Depending on the topology of such interconnections, multiple BD-RIS architectures have been proposed, such as group-/fully-connected RIS \cite{she20,li22-3}, forest-/tree-connected RIS \cite{ner23-1}, and non-diagonal RIS based on dynamic interconnections \cite{li22,wan23b}.
The improved flexibility of BD-RIS enables higher performance gains than conventional RIS at the cost of a moderate increase in the circuit complexity \cite{ner23-2} and channel estimation overhead \cite{li23-3}.
Recent works on BD-RIS optimization have shown the superiority of BD-RIS in multi-antenna \cite{ner22,san23,ner21}, multi-user \cite{fan22,fan23}, and \gls{mec} \cite{mah23} systems.
Furthermore, BD-RIS proved to be especially beneficial in the presence of mutual coupling \cite{li23-4,ner23-3}.

Beyond achieving higher performance gains, BD-RIS working in the so-called hybrid mode enables both the reflection and transmission of the incident signal, allowing full-space coverage \cite{li22-1}.
To further improve the performance while preserving full-space coverage, multi-sector BD-RIS has been introduced, where the elements are divided into multiple sectors, each covering a narrow region of space \cite{li22-2}.
BD-RIS with reflective and transmissive capabilities has proved to improve the capacity and sensing precision in \gls{dfrc} systems \cite{wan23a} and to improve the sum rate and enlarge the coverage in \gls{rsma} systems \cite{li23-2}.

Although BD-RIS working in hybrid mode and multi-sector BD-RIS have shown great coverage improvements \cite{li22-1}-\cite{li23-2}, an RIS effectively enhances the link between a transmitter and a receiver only when placed close to one of them \cite{wu21}.
For this reason, an RIS can be particularly beneficial only to users in its proximity, yielding reduced coverage capabilities.
This reduced coverage limitation arises from the fact that RIS has been commonly regarded as a localized array of scattering elements, here denoted as \textit{localized RIS}.
In this study, we overcome this limitation by investigating RIS made of a distributed array of elements, in which the inter-element distance can be much longer than the wavelength, referred to as \textit{distributed RIS}.
Our focus will extend specifically to distributed BD-RIS, as they are expected to achieve particularly enhanced performance due to the presence of tunable impedance components interconnecting the RIS elements.
These interconnections have the potential to effectively guide the \gls{em} signal toward the receiver by enabling its propagation within the RIS.
Distributed BD-RIS could be implemented by integrating the RIS elements into long cables, forming long linear arrays.
Specifically, distributed BD-RIS could drive the \gls{em} signal from a fixed \gls{bs} to the mobile \glspl{ue} within a coverage area proportional to their array length in highly obstructed indoor environments, e.g., smart factories, as well as in urban outdoor settings\footnote{Note that the concept of distributed BD-RIS is related to \gls{mimo} arrays with antenna spacing larger than half-wavelength and to ``radio stripes'' \cite{sha21}, although they are fundamentally different as distributed BD-RIS is envisioned to be purely passive.
Furthermore, a distributed RIS differs from multiple RISs as multiple RISs are multiple, separate, localized RISs deployed at different points in the environment.}.

The investigation of distributed BD-RIS requires modeling the interconnections in the BD-RIS architecture by accounting for their losses, which is a challenge overlooked in previous literature on BD-RIS.
In previous works \cite{li23-1}-\cite{li23-2}, BD-RIS has been modeled through a lumped-element circuit model, valid when the physical dimensions of the RIS reconfigurable impedance network are much smaller than the wavelength.
According to the lumped-element model, voltages and currents are constant along the interconnections, which simplifies the BD-RIS modeling and optimization.
However, the length of the interconnections between the RIS elements in BD-RIS can be a considerable fraction of the wavelength in localized BD-RIS, or many wavelengths in distributed BD-RIS.
For example, even when the inter-element distance is half-wavelength, the interconnections are at least half-wavelength long.
In the case the interconnections are not much smaller than the wavelength, two critical effects take place.
First, voltages and currents vary in phase along the interconnections.
Second, voltages and currents may vary in magnitude along the interconnections due to losses.
Thus, it is crucial to develop a BD-RIS model accounting for these two effects, as they may lead to undesired changes of phase and dissipation of power within the BD-RIS circuit.
To this end, in this study, we model and optimize localized and distributed BD-RIS by characterizing its interconnections through transmission line theory.
Our contributions are summarized as follows.

\textit{First}, we propose the concept of distributed RIS, in which the antenna elements are not localized in a specific site but distributed over a wide region.
We analyze their scaling laws and quantify the gain of distributed RIS over localized RIS, and of distributed BD-RIS over distributed conventional RIS.
The derived scaling laws show that lossless distributed BD-RIS can offer substantial gains over distributed conventional RIS and localized BD-RIS, as high as several orders of magnitude.
These gains are enabled by the interconnections present in BD-RIS, which allow the guided propagation of the \gls{em} signal within the BD-RIS.

\textit{Second}, we model BD-RIS (localized and distributed) with lossy interconnections by using transmission line theory.
Specifically, we model the interconnections present in the BD-RIS circuit topology as tunable impedance components in series with lossy transmission lines.
Thus, we derive the expressions of the elements of the BD-RIS admittance matrix as functions of the tunable impedance components and the transmission line parameters.

\textit{Third}, since it is difficult to derive engineering insights from the derived BD-RIS model, we obtain three simplified models.
First, by making specific assumptions on the interconnection lengths, we show how losses impact the BD-RIS admittance matrix elements.
Second, by assuming lossless interconnections, we illustrate how the presence of long interconnections affects the lossless BD-RIS model.
Third, by assuming lossless interconnections with specific lengths, it is shown how the derived model boils down to the BD-RIS model widely considered in previous literature.

\textit{Fourth}, we optimize lossy BD-RIS based on the proposed models and assess its performance, compared with conventional RIS.
Numerical results show that the performance of localized BD-RIS is only slightly impacted by losses given the short interconnection lengths.
Besides, the performance of distributed RIS, even with losses, can achieve orders of magnitude of gains over conventional RIS.

\textit{Organization}:
In Section~\ref{sec:system}, we introduce the localized and distributed BD-RIS-aided system model.
In Section~\ref{sec:scaling-laws}, we derive and compare the scaling laws of localized and distributed BD-RIS.
In Section~\ref{sec:gen-model}, we model BD-RIS with lossy interconnections through transmission line theory.
In Section~\ref{sec:sim-model}, we derive three simplified models for lossy BD-RIS.
In Section~\ref{sec:optimization}, we optimize BD-RIS characterized with the proposed models.
In Section~\ref{sec:results}, we present the numerical results.
Finally, Section~\ref{sec:conclusion} concludes this work.

\textit{Notation}:
Vectors and matrices are denoted with bold lower and bold upper letters, respectively.
Scalars are represented with letters not in bold font.
$\Re\{a\}$, $\Im\{a\}$, and $\vert a\vert$ refer to the real part, the imaginary part, and the absolute value of a complex scalar $a$, respectively.
$[\mathbf{a}]_{i}$ and $\Vert\mathbf{a}\Vert$ refer to the $i$th element and $l_{2}$-norm of a vector $\mathbf{a}$, respectively.
$\mathbf{A}^T$, $\mathbf{A}^H$, and $[\mathbf{A}]_{i,j}$ refer to the transpose, conjugate transpose, and $(i,j)$th element of a matrix $\mathbf{A}$, respectively.
$\mathbb{R}$, $\mathbb{R}_{+}$, $\mathbb{R}_{*}$, and $\mathbb{C}$ denote the real, positive real, non-zero real, and complex number sets, respectively.
$j=\sqrt{-1}$ denotes the imaginary unit.
$\mathbf{0}$ and $\mathbf{I}$ denote an all-zero matrix and an identity matrix with appropriate dimensions, respectively.
$\mathcal{CN}(\mathbf{0},\mathbf{I})$ denotes the distribution of a circularly symmetric complex Gaussian random vector with mean vector $\mathbf{0}$ and covariance matrix $\mathbf{I}$ and $\sim$ stands for ``distributed as''.
diag$(a_1,\ldots,a_N)$ refers to a diagonal matrix with diagonal elements being $a_1,\ldots,a_N$, and diag$(\mathbf{a})$ refers to a diagonal matrix with diagonal elements being the vector $\mathbf{a}$.

\section{Localized and Distributed BD-RIS-Aided System Model}
\label{sec:system}

Consider a \gls{mimo} system between an $N_T$-antenna transmitter and an $N_R$-antenna receiver aided by an $N$-element BD-RIS.
The $N$ elements of the BD-RIS are connected to a $N$-port reconfigurable impedance network, with scattering matrix $\boldsymbol{\Theta}\in\mathbb{C}^{N\times N}$.
Given the matrix $\boldsymbol{\Theta}$, the wireless channel $\mathbf{H}\in\mathbb{C}^{N_R\times N_T}$ writes as
\begin{equation}
\mathbf{H}=\mathbf{H}_{RT}+\mathbf{H}_{R}\boldsymbol{\Theta}\mathbf{H}_{T},\label{eq:H}
\end{equation}
where $\mathbf{H}_{RT}\in\mathbb{C}^{N_{R}\times N_{T}}$, $\mathbf{H}_{R}\in\mathbb{C}^{N_{R}\times N}$, and $\mathbf{H}_{T}\in\mathbb{C}^{N\times N_{T}}$ refer to the channels from transmitter to receiver, RIS to receiver, and transmitter to RIS, respectively \cite{she20}.

In related literature, an RIS has been regarded as an antenna array with an inter-element distance comparable with the wavelength.
We refer to this type of RIS as localized RIS, as its elements are all localized in a specific site.
With a localized RIS, the channels $\mathbf{H}_{R}$ and $\mathbf{H}_{T}$ have been commonly modeled as $\mathbf{H}_{R}^{\text{Loc}}=\sqrt{\rho_{R}}\widetilde{\mathbf{H}}_{R}$ and $\mathbf{H}_{T}^{\text{Loc}}=\sqrt{\rho_{T}}\widetilde{\mathbf{H}}_{T}$, where $\rho_{R}$ and $\rho_{T}$ are the path-gains while  $\widetilde{\mathbf{H}}_{R}\in\mathbb{C}^{N_{R}\times N}$ and $\widetilde{\mathbf{H}}_{T}\in\mathbb{C}^{N\times N_{T}}$ are the small-scale fading effects.
Opposed to localized RISs, we investigate in this study the performance of distributed RISs, i.e., RISs whose elements are distributed over a wide region, with an inter-element distance that can be much longer than the wavelength (for example, greater than $10\lambda$, where $\lambda$ is the wavelength).
In the presence of a distributed RIS, the channel $\mathbf{H}_{R}$ is expressed as $\mathbf{H}_{R}^{\text{Dis}}=\widetilde{\mathbf{H}}_{R}\mathbf{R}_{R}^{1/2}$, where $\mathbf{R}_{R}=\text{diag}(\boldsymbol{\rho}_R)$, with $\boldsymbol{\rho}_R\in\mathbb{R}^{N\times 1}$ introduced such that $[\boldsymbol{\rho}_{R}]_{n}$ is the path-gain from the $n$th RIS element to the receiver.
Similarly, we have $\mathbf{H}_{T}^{\text{Dis}}=\mathbf{R}_{T}^{1/2}\widetilde{\mathbf{H}}_{T}$, where $\mathbf{R}_{T}=\text{diag}(\boldsymbol{\rho}_T)$, with $\boldsymbol{\rho}_T\in\mathbb{R}^{N\times 1}$ introduced such that $[\boldsymbol{\rho}_{T}]_{n}$ is the path-gain from the transmitter to the $n$th RIS element\footnote{As the channel expression for localized and distributed is the same, as given by \eqref{eq:H}, the channel can be estimated in the presence of a distributed RIS with the same protocols proposed for localized RIS \cite{you20,li23-3}.}.

In this BD-RIS-aided system, we denote the transmitted signal as $\mathbf{x}=\mathbf{w}s\in\mathbb{C}^{N_{T}\times1}$, where $\mathbf{w}\in\mathbb{C}^{N_{T}\times1}$ is the precoding vector subject to $\Vert\mathbf{w}\Vert=1$, and $s\in\mathbb{C}$ is the transmitted symbol with average power $P_{T}=\text{E}[\vert s\vert^2]$.
Denoting the received signal as $\mathbf{y}\in\mathbb{C}^{N_{R}\times1}$, we have $\mathbf{y}=\mathbf{H}\mathbf{x}+\mathbf{n}$, where $\mathbf{n}\in\mathbb{C}^{N_{R}\times1}$ is the \gls{awgn} at the receiver with power $\sigma_{n}^{2}$.
By using a combining vector $\mathbf{g}\in\mathbb{C}^{1\times N_{R}}$ subject to $\Vert\mathbf{g}\Vert=1$, the signal used for detection $\hat{s}=\mathbf{g}\mathbf{y}$ can be expressed as
\begin{equation}
\hat{s}=\mathbf{g}\left(\mathbf{H}_{RT}+\mathbf{H}_{R}\boldsymbol{\Theta}\mathbf{H}_{T}\right)\mathbf{w}s+\tilde{n},\label{eq:z}
\end{equation}
where $\tilde{n}=\mathbf{g}\mathbf{n}$ is the \gls{awgn} with power $\sigma_{n}^{2}$.
Thus, when reconfiguring the BD-RIS, $\boldsymbol{\Theta}$ is optimized jointly with $\mathbf{g}$ and $\mathbf{w}$ to maximize the received signal power, given by
\begin{equation}
P_R=P_T\left\vert\mathbf{g}\left(\mathbf{H}_{RT}+\mathbf{H}_{R}\boldsymbol{\Theta}\mathbf{H}_{T}\right)\mathbf{w}\right\vert^2.\label{eq:PR}
\end{equation}
In the following, we analyze and compare the received signal power scaling laws of localized and distributed BD-RIS.

\section{Scaling Laws}
\label{sec:scaling-laws}

To compare the fundamental limits of localized and distributed BD-RIS, we derive their received signal power scaling laws.
For simplicity, we consider a \gls{siso} system, i.e., $N_R=1$ and $N_T=1$, with obstructed direct channel, i.e., $\mathbf{H}_{RT}=\mathbf{0}$, and lossless BD-RIS, i.e., with unitary $\boldsymbol{\Theta}$.
With no loss of generality, we assume transmit power $P_T=1$, such that the received signal power writes as $P_R=\vert\mathbf{h}_{R}\boldsymbol{\Theta}\mathbf{h}_{T}\vert^2$.
We model the path-gain of the channels $\mathbf{h}_{R}$ and $\mathbf{h}_{T}$ through the distance-dependent model and their small-scale fading as \gls{iid} Rayleigh distributed.
Specifically, for localized BD-RIS, we have $\mathbf{h}_{i}^{\text{Loc}}\sim\mathcal{CN}(\mathbf{0},\rho_{i}\mathbf{I})$, with $\rho_{i}=C_{0}d_{i}^{-a}$, for $i\in\{R,T\}$, where $C_{0}$ refers to the path-gain at the reference distance 1~m, $a$ is the path-loss exponent, and $d_{R}$ (resp. $d_{T}$) is the distance between the RIS and the receiver (resp. the transmitter).
Similarly, for distributed BD-RIS, we have $\mathbf{h}_{i}^{\text{Dis}}\sim\mathcal{CN}(\mathbf{0},\mathbf{R}_{i})$, with $\mathbf{R}_{i}=\text{diag}(\boldsymbol{\rho}_{i})$, $[\boldsymbol{\rho}_{i}]_{n}=C_{0}[\mathbf{d}_{i}]_{n}^{-a}$, for $i\in\{R,T\}$, where $[\mathbf{d}_{R}]_{n}$ (resp. $[\mathbf{d}_{T}]_{n}$) is the distance between the $n$th RIS element and the receiver (resp. the transmitter).

In the following, we study the localized and distributed versions of single- and fully-connected RIS, which are the least and the most complex BD-RIS architectures, respectively.
Single-connected RIS is the conventional RIS characterized by a unitary and diagonal scattering matrix, i.e., $\boldsymbol{\Theta}=\text{diag}(e^{j\theta_1},\ldots,e^{j\theta_N})$, with $\theta_n\in[0,2\pi)$, for $n=1,\ldots,N$.
Besides, fully-connected RIS is the BD-RIS architecture with the highest flexibility as it is characterized by an arbitrary symmetric and unitary scattering matrix \cite{she20}.

For localized single-connected RIS, it is possible to optimize the phase shifts $\theta_n$ to obtain a received signal power given by
\begin{equation}
P_R^{\text{Loc-SC}}
=\left(\sum_{n=1}^{N}\left\vert\left[\mathbf{h}_{R}^{\text{Loc}}\right]_n\right\vert\left\vert\left[\mathbf{h}_{T}^{\text{Loc}}\right]_n\right\vert\right)^2.\label{eq:P-Loc-SC-random}
\end{equation}
Furthermore, the expected value of \eqref{eq:P-Loc-SC-random} can be derived by exploiting the \gls{iid} channels assumption, and that $\text{E}[\vert[\mathbf{h}_{i}^{\text{Loc}}]_{n}\vert^2]=\rho_{i}$ and $\text{E}[\vert[\mathbf{h}_{i}^{\text{Loc}}]_{n}\vert]=\frac{\sqrt{\pi}}{2}\sqrt{\rho_{i}}$, for $i\in\{R,T\}$ and $n=1,\ldots,N$, because of the moments of the chi distribution with $2$ degrees of freedom.
Thus, it is possible to show that
\begin{align}
\text{E}\left[P_R^{\text{Loc-SC}}\right]
&=\left(N+\frac{\pi^2}{16}N\left(N-1\right)\right)C_{0}^2d_{R}^{-a}d_{T}^{-a},\label{eq:P-Loc-SC}
\end{align}
giving the scaling law of localized single-connected RIS following the steps in \cite{she20}.

In the case of localized fully-connected RIS, it has been shown that it is always possible to optimize the RIS to achieve a maximum received signal power given by
\begin{equation}
P_R^{\text{Loc-FC}}
=\left\Vert\mathbf{h}_{R}^{\text{Loc}}\right\Vert^2\left\Vert\mathbf{h}_{T}^{\text{Loc}}\right\Vert^2,
\end{equation}
whose expected value writes as
\begin{align}
\text{E}\left[P_R^{\text{Loc-FC}}\right]
&=N^2C_{0}^2d_{R}^{-a}d_{T}^{-a},\label{eq:P-Loc-FC}
\end{align}
following the \gls{iid} channels assumption and the moments of the chi distribution with $2$ degrees of freedom \cite{she20}.

For distributed single-connected RIS, the achievable received signal power writes as
\begin{equation}
P_R^{\text{Dis-SC}}
=\left(\sum_{n=1}^{N}\left\vert\left[\mathbf{h}_{R}^{\text{Dis}}\right]_n\right\vert\left\vert\left[\mathbf{h}_{T}^{\text{Dis}}\right]_n\right\vert\right)^2,
\end{equation}
similarly to \eqref{eq:P-Loc-SC-random}.
By exploiting the \gls{iid} channels assumption, and that $\text{E}[\vert[\mathbf{h}_{i}^{\text{Dis}}]_{n}\vert^2]=[\boldsymbol{\rho}_{i}]_{n}$ and $\text{E}[\vert[\mathbf{h}_{i}^{\text{Dis}}]_{n}\vert]=\frac{\sqrt{\pi}}{2}\sqrt{[\boldsymbol{\rho}_{i}]_{n}}$, for $i\in\{R,T\}$ and $n=1,\ldots,N$, it is possible to show that
\begin{multline}
\text{E}\left[P_R^{\text{Dis-SC}}\right]
=C_{0}^2\left(\sum_{n=1}^{N}\left(\left[\mathbf{d}_{R}\right]_{n}\left[\mathbf{d}_{T}\right]_{n}\right)^{-a}\right.\\
\left.+\frac{\pi^2}{16}\sum_{n\neq m}\left(\left[\mathbf{d}_{R}\right]_{n}\left[\mathbf{d}_{T}\right]_{n}\left[\mathbf{d}_{R}\right]_{m}\left[\mathbf{d}_{T}\right]_{m}\right)^{-\frac{a}{2}}\right),\label{eq:P-Dis-SC}
\end{multline}
giving the scaling law of distributed single-connected RIS.

Finally, in the case of distributed fully-connected RIS, the achievable received signal power is
\begin{equation}
P_R^{\text{Dis-FC}}
=\left\Vert\mathbf{h}_{R}^{\text{Dis}}\right\Vert^2\left\Vert\mathbf{h}_{T}^{\text{Dis}}\right\Vert^2,
\end{equation}
with expected value
\begin{align}
\text{E}\left[P_R^{\text{Dis-FC}}\right]
&=C_{0}^2\sum_{n=1}^{N}\left[\mathbf{d}_{R}\right]_{n}^{-a}\sum_{n=1}^{N}\left[\mathbf{d}_{T}\right]_{n}^{-a},\label{eq:P-Dis-FC}
\end{align}
following again the \gls{iid} channels assumption and the moments of the chi distribution with $2$ degrees of freedom
\footnote{Under \gls{los} small-scale fading, i.e., $[\mathbf{h}_{i}^{\text{Loc}}]_{n}=\sqrt{\rho_i}e^{j\phi_{i,n}}$ and $[\mathbf{h}_{i}^{\text{Dis}}]_{n}=\sqrt{[\boldsymbol{\rho}_{i}]_n}e^{j\psi_{i,n}}$, for $i\in\{R,T\}$ and $n=1,\ldots,N$, it is possible to prove that $\text{E}[P_R^{\text{Loc-FC}}]$ and $\text{E}[P_R^{\text{Dis-FC}}]$ remain as in \eqref{eq:P-Loc-FC} and \eqref{eq:P-Dis-FC}, $\text{E}[P_R^{\text{Loc-SC}}]=\text{E}[P_R^{\text{Loc-FC}}]$, and $\text{E}[P_R^{\text{Dis-SC}}]=C_0^2(\sum_{n}([\mathbf{d}_{R}]_{n}[\mathbf{d}_{T}]_{n})^{-a/2})^2$.
Given the similarity between the scaling laws under Rayleigh and \gls{los} fading, the following insights on localized and distributed RIS obtained considering Rayleigh fading are also valid under more general Rician fading.}.

\begin{figure}[t]
\centering
\subfigure[]{
\includegraphics[width=0.46\textwidth]{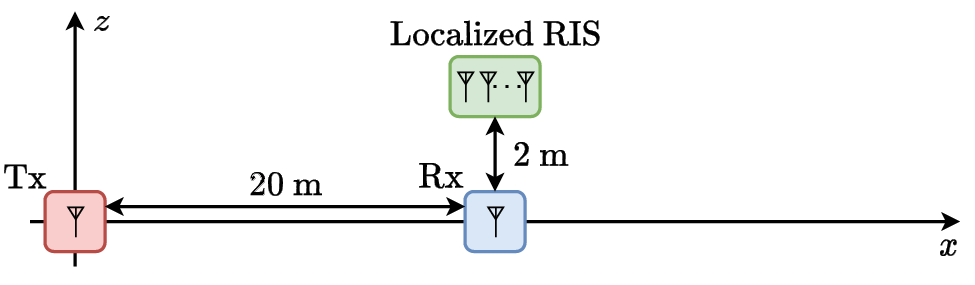}
\label{fig:deployment-loc}
}
\subfigure[]{
\includegraphics[width=0.46\textwidth]{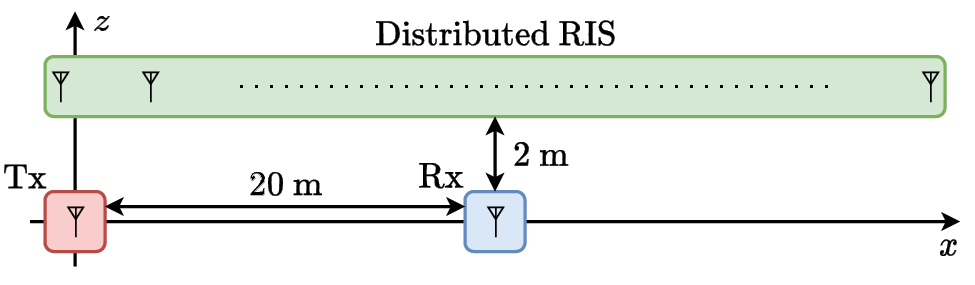}
\label{fig:deployment-dis}
}
\caption{(a) Localized and (b) distributed RIS-aided communication system.}
\label{fig:deployment}
\end{figure}

In the following, we compare the four scaling laws derived in \eqref{eq:P-Loc-SC}, \eqref{eq:P-Loc-FC}, \eqref{eq:P-Dis-SC}, and \eqref{eq:P-Dis-FC} by analyzing four gains:
\textit{i)} the gain of fully- over single-connected localized RIS,
\textit{ii)} the gain of fully- over single-connected distributed RIS,
\textit{iii)} the gain of distributed over localized single-connected RIS, and
\textit{iv)} the gain of distributed over localized fully-connected RIS.
First, we carry out a theoretical discussion valid for any values of the distances $d_R$, $d_T$, $\mathbf{d}_R$, and $\mathbf{d}_T$.
Second, to numerically quantify these gains, we consider a three-dimensional coordinate system where the transmitter and receiver are located at $(0,0,0)$ and $(20,0,0)$ in meters (m), respectively.
In the case of localized RIS, the RIS is located in $(20,0,2)$, as shown in Fig.~\ref{fig:deployment-loc}.
In the case of distributed RIS, the RIS elements are spread out over a long line.
We consider a linear architecture as it can be implemented by embedding the RIS elements and their interconnections into a cable, but the proposed distributed RIS is not limited to linear architectures.
Specifically, the RIS is a \gls{ula} with elements uniformly placed between $(0,0,2)$ and $(40,0,2)$, as shown in Fig.~\ref{fig:deployment-dis}.
To obtain more general insights, we also consider the same systems as in Figs.~\ref{fig:deployment-loc} and \ref{fig:deployment-dis}, but where the location of the receiver varies within an area of interest.

The aforementioned gains are theoretically analyzed by exploiting the concept of generalized mean in the following subsections.
The generalized mean with exponent $p\in\mathbb{R}_{*}$ of the elements of a vector $\boldsymbol{\xi}\in\mathbb{R}_{+}^{N\times1}$ is denoted as $M_p(\boldsymbol{\xi})=(\sum_{n=1}^{N}[\boldsymbol{\xi}]_{n}^{p}/N)^{1/p}$.
In particular, we make use of the following result.
\begin{proposition}
Given $N$ different positive real number collected in a vector $\boldsymbol{\xi}\in\mathbb{R}_{+}^{N\times1}$, the generalized mean with exponent $p\in\mathbb{R}_{*}$ of the elements of $\boldsymbol{\xi}$ $M_p(\boldsymbol{\xi})$ satisfy
\begin{equation}
\emph{min}\left(\boldsymbol{\xi}\right)<M_{-p}\left(\boldsymbol{\xi}\right)<\sqrt[p]{N}\emph{min}\left(\boldsymbol{\xi}\right).
\end{equation}
\label{pro:bounds}
\end{proposition}
\begin{proof}
To prove the two bounds, we rewrite $M_{-p}\left(\boldsymbol{\xi}\right)$ as
\begin{equation}
M_{-p}\left(\boldsymbol{\xi}\right)=\sqrt[p]{\frac{N}{\sum_{n=1}^N\frac{1}{\left[\boldsymbol{\xi}\right]_{n}^p}}}.
\end{equation}
First, since $\sum_{n=1}^N\frac{1}{\left[\boldsymbol{\xi}\right]_{n}^p}<\frac{N}{\text{min}\left(\boldsymbol{\xi}\right)^p}$, we have $M_{-p}\left(\boldsymbol{\xi}\right)>\text{min}\left(\boldsymbol{\xi}\right)$, with equality achieved if and only if all entries of $\boldsymbol{\xi}$ are equal.
Second, since $\sum_{n=1}^N\frac{1}{\left[\boldsymbol{\xi}\right]_{n}^p}>\frac{1}{\text{min}\left(\boldsymbol{\xi}\right)^p}$, we have $M_{-p}\left(\boldsymbol{\xi}\right)<\sqrt[p]{N}\text{min}\left(\boldsymbol{\xi}\right)$.
\end{proof}

\subsection{Gain of Fully- over Single-Connected Localized RIS}
\label{sec:G-Loc}

The gain of fully-connected over single-connected localized RIS, defined as $\mathcal{G}^{\text{Loc}}=\text{E}[P_R^{\text{Loc-FC}}]/\text{E}[P_R^{\text{Loc-SC}}]$, writes as
\begin{equation}
\mathcal{G}^{\text{Loc}}
=\frac{N}{1+\frac{\pi^2}{16}\left(N-1\right)},
\end{equation}
following \eqref{eq:P-Loc-SC} and \eqref{eq:P-Loc-FC}, in agreement with \cite{she20}.
Remarkably, we observe that $\mathcal{G}^{\text{Loc}}$, which depends only on $N$, is lower bounded by $\mathcal{G}^{\text{Loc}}\geq1$ and upper bounded by $\mathcal{G}^{\text{Loc}}<\frac{16}{\pi^2}\approx1.62$.
Thus, localized fully-connected RIS is beneficial over localized single-connected RIS, enabling a gain of at most $62\%$ in a \gls{siso} system \cite{she20}.

\subsection{Gain of Fully- over Single-Connected Distributed RIS}

We define the gain of fully-connected over single-connected distributed RIS as $\mathcal{G}^{\text{Dis}}=\text{E}[P_R^{\text{Dis-FC}}]/\text{E}[P_R^{\text{Dis-SC}}]$.
By noticing that \eqref{eq:P-Dis-SC} is upper bounded by
\begin{align}
\text{E}\left[P_R^{\text{Dis-SC}}\right]
&<C_{0}^2\left(\sum_{n=1}^{N}\left(\left[\mathbf{d}_{R}\right]_{n}\left[\mathbf{d}_{T}\right]_{n}\right)^{-\frac{a}{2}}\right)^2,\label{eq:P-Dis-SC-UB}
\end{align}
we obtain that
\begin{align}
\mathcal{G}^{\text{Dis}}
&>\frac{\sum_{n=1}^{N}\left[\mathbf{d}_{R}\right]_{n}^{-a}
\sum_{n=1}^{N}\left[\mathbf{d}_{T}\right]_{n}^{-a}}{\left(\sum_{n=1}^{N}\left(\left[\mathbf{d}_{R}\right]_{n}\left[\mathbf{d}_{T}\right]_{n}\right)^{-\frac{a}{2}}\right)^2}\\
&=\left(\frac{M_{-\frac{a}{2}}\left(\mathbf{d}_{R}\odot\mathbf{d}_{T}\right)}{M_{-a}\left(\mathbf{d}_{R}\right)M_{-a}\left(\mathbf{d}_{T}\right)}\right)^{a},
\end{align}
where $M_p(\boldsymbol{\xi})=(\sum_{n=1}^{N}[\boldsymbol{\xi}]_{n}^{p}/N)^{1/p}$ denotes the generalized mean of the elements of $\boldsymbol{\xi}\in\mathbb{R}_{+}^{N\times1}$ with exponent $p\in\mathbb{R}_{*}$, and $\odot$ denotes the Hadamard product.
Thus, we have $\mathcal{G}^{\text{Dis}}\geq 1$ from the Cauchy-Schwarz inequality.
Furthermore, since $\text{min}(\boldsymbol{\xi})<M_{-p}(\boldsymbol{\xi})<\sqrt[p]{N}\text{min}(\boldsymbol{\xi})$ following Proposition~\ref{pro:bounds}, we have
\begin{align}
\mathcal{G}^{\text{Dis}}
&>\left(\frac{\text{min}\left(\mathbf{d}_{R}\odot\mathbf{d}_{T}\right)}{\sqrt[a]{N^2}\text{min}\left(\mathbf{d}_{R}\right)\text{min}\left(\mathbf{d}_{T}\right)}\right)^{a},\label{eq:G-Dis-LB}
\end{align}
indicating that $\mathcal{G}^{\text{Dis}}$ increases exponentially with $a$ and can be significantly high when both $\text{min}\left(\mathbf{d}_{R}\right)$ and $\text{min}\left(\mathbf{d}_{T}\right)$ are reduced, i.e., there is at least an RIS element close to the receiver and another one close to the transmitter.


We numerically quantify $\mathcal{G}^{\text{Dis}}$ in Fig.~\ref{fig:G1-Dis}, where we consider the system represented in Fig.~\ref{fig:deployment-dis} with different values of path-loss exponent $a$ and number of RIS elements $N$.
From Fig.~\ref{fig:G1-Dis}, we observe that $\mathcal{G}^{\text{Dis}}$ can be as high as several orders of magnitude, differently from $\mathcal{G}^{\text{Loc}}$, which is limited to $1.62$.
Fig.~\ref{fig:G1-Dis} confirms that $\mathcal{G}^{\text{Dis}}$ grows exponentially with $a$, and shows that the gain is only slightly affected by $N$\footnote{With $a=4$ and $N=64$, the received signal power offered by a distributed fully-connected RIS is given by $P_R^{\text{Dis-FC}}
=\Vert\mathbf{h}_{R}^{\text{Dis}}\Vert^2\Vert\mathbf{h}_{T}^{\text{Dis}}\Vert^2
=-72$~dBW.
Note that a even higher performance can be achieved by a distributed \gls{mimo} transmitter having $N$ antennas in the locations of the $N$ RIS elements.
In this case, \gls{mrt} would give a received signal power of $P_R^{\text{Dis-MIMO}}
=\Vert\mathbf{h}_{R}^{\text{Dis}}\Vert^2
=-35$~dBW.
However, this improvement would require a much-increased hardware cost and power consumption, since a \gls{rf} chain per element would be needed.}.

Besides, in Fig.~\ref{fig:G2-Dis}, we report the gain $\mathcal{G}^{\text{Dis}}$ for different locations of the receiver, located in $(x,y,0)$, where $x\in[-10,70]$ and $y\in[-40,40]$, with $a=4$ and $N=64$.
From Fig.~\ref{fig:G2-Dis}, we observe that $\mathcal{G}^{\text{Dis}}$ is particularly high when the receiver is close to the distributed RIS and at the same time far from the transmitter.
As visible from $\eqref{eq:G-Dis-LB}$, $\mathcal{G}^{\text{Dis}}$ assumes high values when both the transmitter and the receiver are close to the distributed RIS.
Thus, the high gains in Fig.~\ref{fig:G2-Dis} can be observed also for different locations of the transmitter, as long as it is located in proximity to an RIS element.
Conversely, $\mathcal{G}^{\text{Dis}}$ decreases, approaching $1$, when both the transmitter and the receiver are located far from all the RIS elements.

\begin{figure*}[t]
\centering
\subfigure[Gain of fully- over single-connected distributed RIS $\mathcal{G}^{\text{Dis}}$.]{
\includegraphics[width=0.29\textwidth]{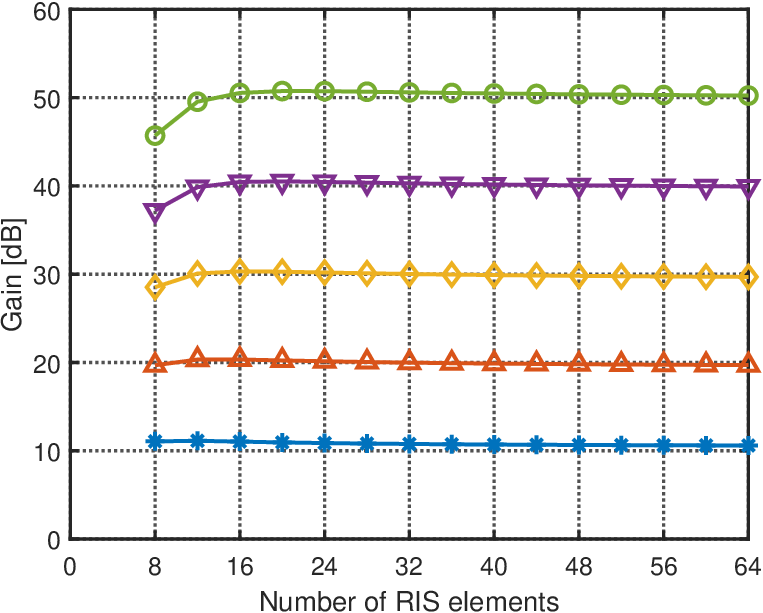}
\label{fig:G1-Dis}
}
\subfigure[Gain of distributed over localized single-connected RIS $\mathcal{G}^{\text{SC}}$.]{
\includegraphics[width=0.29\textwidth]{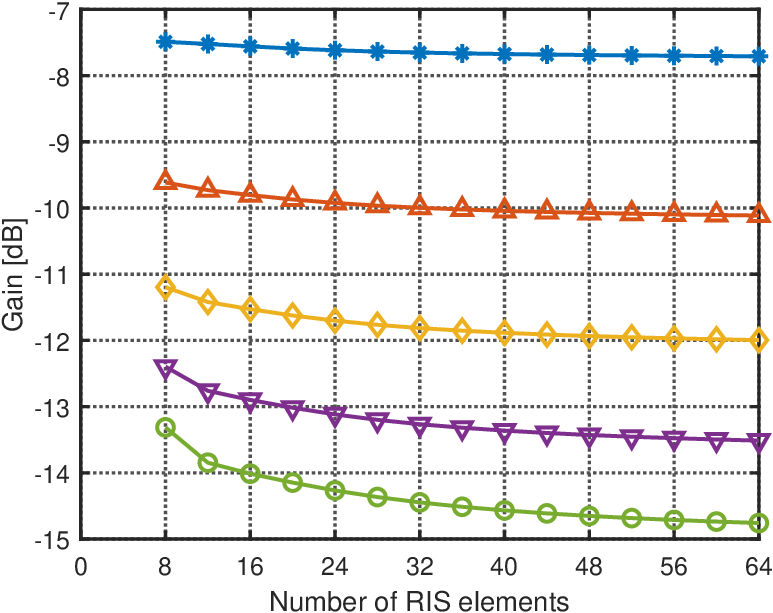}
\label{fig:G1-SC}
}
\subfigure[Gain of distributed over localized fully-connected RIS $\mathcal{G}^{\text{FC}}$.]{
\includegraphics[width=0.29\textwidth]{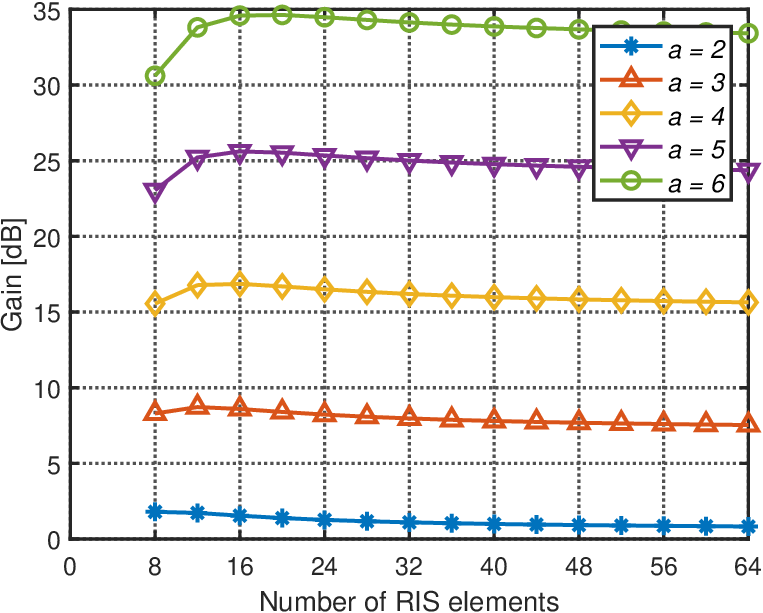}
\label{fig:G1-FC}
}
\caption{(a) $\mathcal{G}^{\text{Dis}}$, (b) $\mathcal{G}^{\text{SC}}$, and (c) $\mathcal{G}^{\text{FC}}$ for different values of path-loss exponent and number of RIS elements.}
\label{fig:G1}
\end{figure*}

\begin{figure*}[t]
\centering
\subfigure[Gain of fully- over single-connected distributed RIS $\mathcal{G}^{\text{Dis}}$.]{
\includegraphics[width=0.29\textwidth]{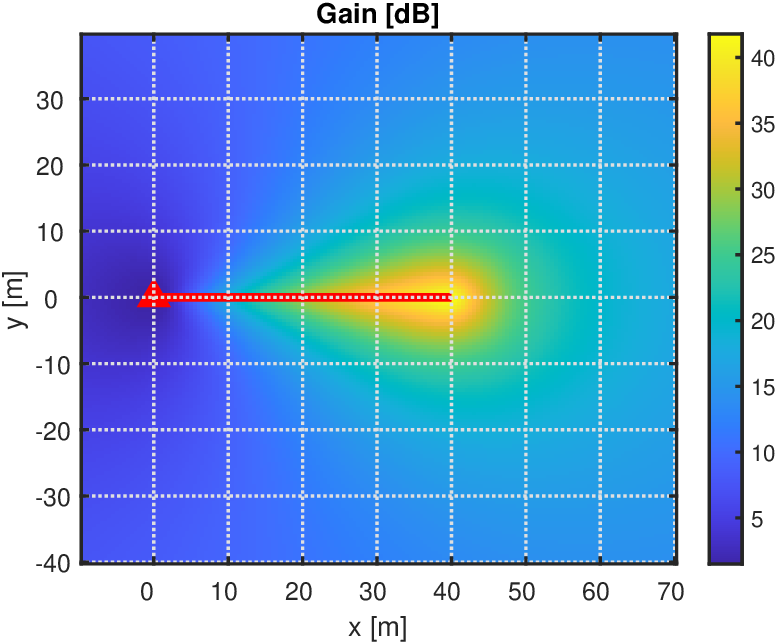}
\label{fig:G2-Dis}
}
\subfigure[Gain of distributed over localized single-connected RIS $\mathcal{G}^{\text{SC}}$.]{
\includegraphics[width=0.29\textwidth]{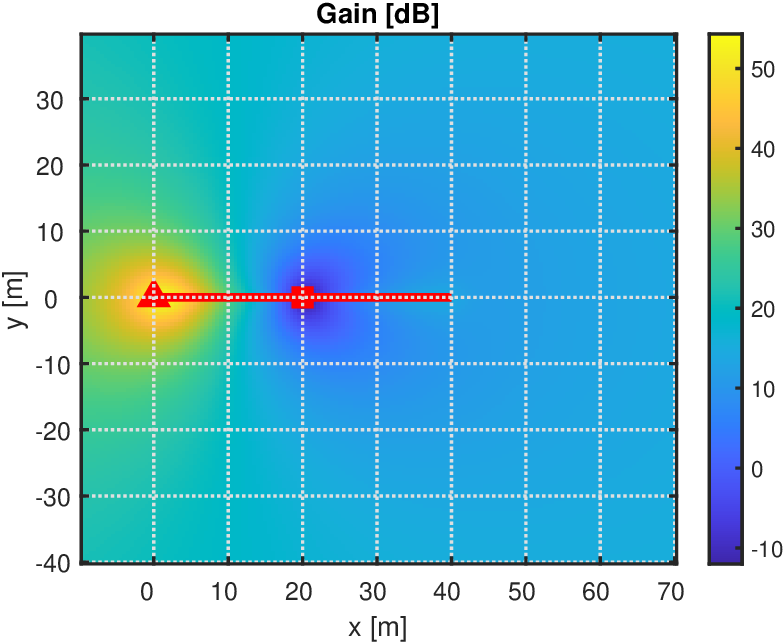}
\label{fig:G2-SC}
}
\subfigure[Gain of distributed over localized fully-connected RIS $\mathcal{G}^{\text{FC}}$.]{
\includegraphics[width=0.29\textwidth]{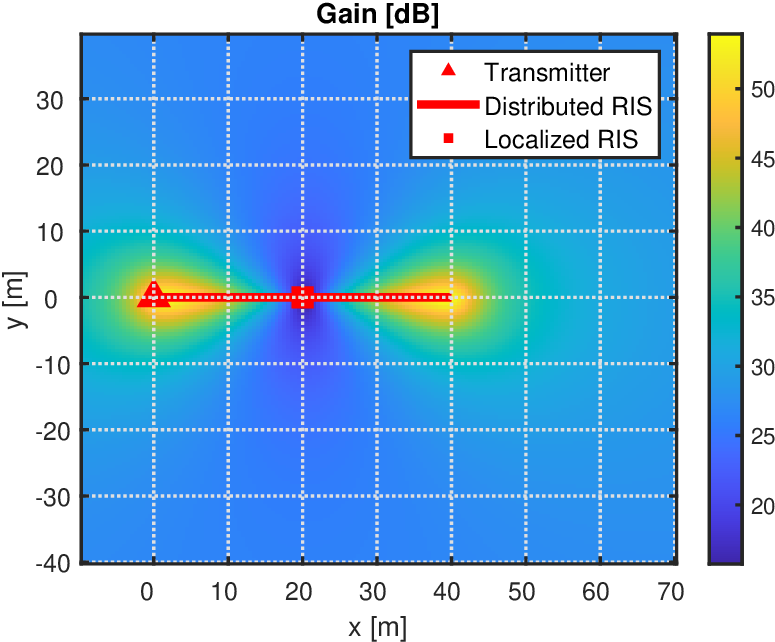}
\label{fig:G2-FC}
}
\caption{(a) $\mathcal{G}^{\text{Dis}}$, (b) $\mathcal{G}^{\text{SC}}$, and (c) $\mathcal{G}^{\text{FC}}$ for different locations of the receiver.}
\label{fig:G2}
\end{figure*}

\subsection{Gain of Distributed over Localized Single-Connected RIS}

Consider now the gain of distributed over localized single-connected RIS, defined as $\mathcal{G}^{\text{SC}}=\text{E}[P_R^{\text{Dis-SC}}]/\text{E}[P_R^{\text{Loc-SC}}]$.
To gain useful insights into $\mathcal{G}^{\text{SC}}$, we approximate this gain by deriving an upper and a lower bound.
First, by noticing that \eqref{eq:P-Dis-SC} is lower bounded by
\begin{align}
\text{E}\left[P_R^{\text{Dis-SC}}\right]
&>C_{0}^2\frac{\pi^2}{16}\left(\sum_{n=1}^{N}\left(\left[\mathbf{d}_{R}\right]_{n}\left[\mathbf{d}_{T}\right]_{n}\right)^{-\frac{a}{2}}\right)^2,
\end{align}
and that \eqref{eq:P-Loc-SC} is upper bounded by
\begin{equation}
\text{E}\left[P_R^{\text{Loc-SC}}\right]
<N^2C_{0}^2d_{R}^{-a}d_{T}^{-a},
\end{equation}
we obtain that
\begin{align}
\mathcal{G}^{\text{SC}}
&>\frac{\frac{\pi^2}{16}\left(\sum_{n=1}^{N}\left(\left[\mathbf{d}_{R}\right]_{n}\left[\mathbf{d}_{T}\right]_{n}\right)^{-\frac{a}{2}}\right)^2}{N^2d_{R}^{-a}d_{T}^{-a}}\\
&=\frac{\pi^2}{16}\left(\frac{d_{R}d_{T}}{M_{-\frac{a}{2}}\left(\mathbf{d}_{R}\odot\mathbf{d}_{T}\right)}\right)^{a}.\label{eq:G-SC-LB}
\end{align}
Second, by recalling \eqref{eq:P-Dis-SC-UB} and noticing that \eqref{eq:P-Loc-SC} is lower bounded by
\begin{equation}
\text{E}\left[P_R^{\text{Loc-SC}}\right]
>\frac{\pi^2}{16}N^2C_{0}^2d_{R}^{-a}d_{T}^{-a},
\end{equation}
we have
\begin{align}
\mathcal{G}^{\text{SC}}
&<\frac{\left(\sum_{n=1}^{N}\left(\left[\mathbf{d}_{R}\right]_{n}\left[\mathbf{d}_{T}\right]_{n}\right)^{-\frac{a}{2}}\right)^2}{\frac{\pi^2}{16}N^2d_{R}^{-a}d_{T}^{-a}}\\
&=\frac{16}{\pi^2}\left(\frac{d_{R}d_{T}}{M_{-\frac{a}{2}}\left(\mathbf{d}_{R}\odot\mathbf{d}_{T}\right)}\right)^{a}.\label{eq:G-SC-UB}
\end{align}
Thus, considering \eqref{eq:G-SC-LB} and \eqref{eq:G-SC-UB}, and noticing that $\pi^2/16 = 0.6169$, we can approximate $\mathcal{G}^{\text{SC}}$ as
\begin{equation}
\mathcal{G}^{\text{SC}}
\approx\left(\frac{d_{R}d_{T}}{M_{-\frac{a}{2}}\left(\mathbf{d}_{R}\odot\mathbf{d}_{T}\right)}\right)^{a}
\triangleq\tilde{\mathcal{G}}^{\text{SC}},\label{eq:G-SC-approx}
\end{equation}
which can be greater or less than one, depending on the positions of the localized and distributed RISs.

Since distributed single-connected RIS is not always beneficial over localized single-connected RIS, we provide intuitive sufficient and necessary conditions for $\tilde{\mathcal{G}}^{\text{SC}}>1$.
First, since $M_{-\frac{a}{2}}(\mathbf{d}_{R}\odot\mathbf{d}_{T})<\sqrt[a]{N^2}\text{min}\left(\mathbf{d}_{R}\odot\mathbf{d}_{T}\right)$ following Proposition~\ref{pro:bounds}, we have
\begin{equation}
\tilde{\mathcal{G}}^{\text{SC}}
>\left(\frac{d_{R}d_{T}}{\sqrt[a]{N^2}\text{min}\left(\mathbf{d}_{R}\odot\mathbf{d}_{T}\right)}\right)^{a},
\end{equation}
and a sufficient condition for $\tilde{\mathcal{G}}^{\text{SC}}>1$ is given by $d_{R}d_{T}>\sqrt[a]{N^2}\text{min}\left(\mathbf{d}_{R}\odot\mathbf{d}_{T}\right)$.
Second, since $M_{-\frac{a}{2}}(\mathbf{d}_{R}\odot\mathbf{d}_{T})>\text{min}(\mathbf{d}_{R}\odot\mathbf{d}_{T})$ following Proposition~\ref{pro:bounds}, we have
\begin{equation}
\tilde{\mathcal{G}}^{\text{SC}}
<\left(\frac{d_{R}d_{T}}{\text{min}\left(\mathbf{d}_{R}\odot\mathbf{d}_{T}\right)}\right)^{a},
\end{equation}
and a necessary condition for $\tilde{\mathcal{G}}^{\text{SC}}>1$ is given by $d_{R}d_{T}>\text{min}(\mathbf{d}_{R}\odot\mathbf{d}_{T})$.

In Fig.~\ref{fig:G1-SC}, we numerically evaluate $\mathcal{G}^{\text{SC}}$ considering the systems represented in Fig.~\ref{fig:deployment}.
Interestingly, we always have $\mathcal{G}^{\text{SC}}<1$, meaning that the distributed single-connected RIS is not beneficial over the localized single-connected RIS in the considered scenario.
This is because the receiver is located close to the localized RIS, yielding a particularly small $d_{R}d_{T}$, which deteriorates the gain, as given in \eqref{eq:G-SC-approx}.

In addition, in Fig.~\ref{fig:G2-SC}, we report $\mathcal{G}^{\text{SC}}$ for different locations of the receiver.
Remarkably, we observe that $\mathcal{G}^{\text{SC}}<1$ only when the receiver is in close proximity to the localized RIS, while it can reach high values when the receiver is far from the localized RIS.
Thus, distributed single-connected RIS offers better coverage than localized single-connected RIS, improving the received signal power over a much wider area.

\subsection{Gain of Distributed over Localized Fully-Connected RIS}

Finally, we consider the gain of distributed over localized fully-connected RIS, defined as $\mathcal{G}^{\text{FC}}=\text{E}[P_R^{\text{Dis-FC}}]/\text{E}[P_R^{\text{Loc-FC}}]$.
From \eqref{eq:P-Loc-FC} and \eqref{eq:P-Dis-FC}, we directly obtain
\begin{align}
\mathcal{G}^{\text{FC}}
&=\frac{\sum_{n=1}^{N}\left[\mathbf{d}_{R}\right]_{n}^{-a}\sum_{n=1}^{N}\left[\mathbf{d}_{T}\right]_{n}^{-a}}{N^2d_{R}^{-a}d_{T}^{-a}}\\
&=\left(\frac{d_{R}d_{T}}{M_{-a}\left(\mathbf{d}_{R}\right)M_{-a}\left(\mathbf{d}_{T}\right)}\right)^{a},
\end{align}
indicating that $\mathcal{G}^{\text{FC}}\geq\mathcal{G}^{\text{SC}}$ by the Cauchy-Schwarz inequality.
Since
\begin{equation}
\mathcal{G}^{\text{FC}}
>\left(\frac{d_{R}d_{T}}{\sqrt[a]{N^2}\text{min}\left(\mathbf{d}_{R}\right)\text{min}\left(\mathbf{d}_{T}\right)}\right)^{a},
\end{equation}
because of Proposition~\ref{pro:bounds}, a sufficient condition for $\mathcal{G}^{\text{FC}}>1$ is $d_{R}d_{T}>\sqrt[a]{N^2}\text{min}\left(\mathbf{d}_{R}\right)\text{min}\left(\mathbf{d}_{T}\right)$, which is more relaxed than the sufficient condition for $\tilde{\mathcal{G}}^{\text{SC}}>1$.
In addition, since
\begin{equation}
\mathcal{G}^{\text{FC}}
<\left(\frac{d_{R}d_{T}}{\text{min}\left(\mathbf{d}_{R}\right)\text{min}\left(\mathbf{d}_{T}\right)}\right)^{a},
\end{equation}
because of Proposition~\ref{pro:bounds}, a necessary condition for $\mathcal{G}^{\text{FC}}>1$ is $d_{R}d_{T}>\text{min}\left(\mathbf{d}_{R}\right)\text{min}\left(\mathbf{d}_{T}\right)$, which is more relaxed than the necessary condition for $\tilde{\mathcal{G}}^{\text{SC}}>1$.

In Fig.~\ref{fig:G1-FC}, we numerically evaluate $\mathcal{G}^{\text{FC}}$ considering the systems represented in Fig.~\ref{fig:deployment}.
Remarkably, we always observe $\mathcal{G}^{\text{FC}}>1$, meaning that the distributed fully-connected RIS always outperforms localized fully-connected RIS in the considered scenario.
Furthermore, Fig.~\ref{fig:G1-FC} shows that $\mathcal{G}^{\text{FC}}$ grows exponentially with $a$, reaching values as high as several orders of magnitude, and is only minimally impacted by $N$.

Besides, we report $\mathcal{G}^{\text{FC}}$ for different locations of the receiver in Fig.~\ref{fig:G2-FC}.
Remarkably, we observe that $\mathcal{G}^{\text{FC}}>15$~dB for all receiver locations, indicating that distributed fully-connected RIS offers better performance than localized fully-connected RIS over the whole considered area.
In particular, $\mathcal{G}^{\text{FC}}$ is significantly high when the receiver is far from the localized RIS while being close to the distributed RIS.

\section{General Lossy BD-RIS Model}
\label{sec:gen-model}

We have introduced the concept of distributed RIS, showing that distributed BD-RIS can enable significant gains over distributed single-connected RIS and localized BD-RIS.
In this section, we study the impact of the losses within the BD-RIS architecture by developing a model for BD-RIS that accounts for lossy interconnections.

We model a BD-RIS with lossy interconnections as an $N$-port network with a circuit topology described as follows.
First, port $m$ is connected to ground through a tunable impedance $Z_m$, for $m=1,\ldots,N$, as represented in Fig.~\ref{fig:model-RIS}(a).
Second, port $m$ and port $n$, if interconnected in the BD-RIS architecture, are interconnected through a tunable impedance $Z_{n,m}$ in series with a transmission line of length $\ell_{n,m}$, for $n,m=1,\ldots,N$, as represented in Fig.~\ref{fig:model-RIS}(b).

\begin{figure}[t]
\centering
\includegraphics[width=0.44\textwidth]{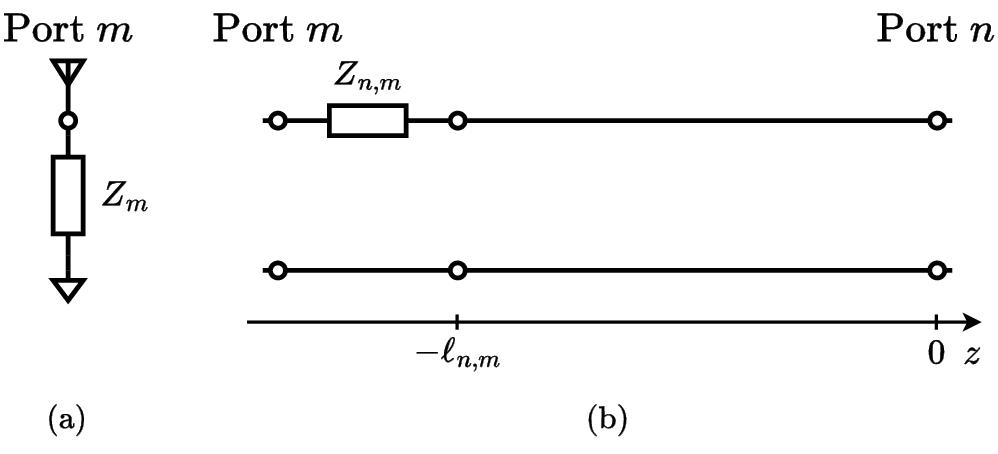}
\caption{(a) Port $m$ connected to ground through $Z_m$ and (b) ports $m$ and $n$ interconnected through $Z_{n,m}$ and a transmission line of length $\ell_{n,m}$.}
\label{fig:model-RIS}
\end{figure}

Given an $N$-port network implemented with this circuit topology, our goal is to characterize its scattering matrix $\boldsymbol{\Theta}$.
Since it is hard to directly derive $\boldsymbol{\Theta}$ from the BD-RIS circuit topology, we characterize the BD-RIS admittance matrix $\mathbf{Y}\in\mathbb{C}^{N\times N}$, which is related to $\boldsymbol{\Theta}$ by 
\begin{equation}
\boldsymbol{\Theta}=\left(\mathbf{I}+Z_{0}\mathbf{Y}\right)^{-1}\left(\mathbf{I}-Z_{0}\mathbf{Y}\right),\label{eq:T(Y)}
\end{equation}
where $Z_{0}$ is the reference impedance, typically set to $Z_0=50$~$\Omega$ \cite[Chapter~4]{poz11}.
According to the definition of $\mathbf{Y}$, each entry $[\mathbf{Y}]_{n,m}$ is obtained by driving port $m$ with a voltage $V_m$, short-circuiting all other ports, measuring the short-circuit current $I_n$ at port $n$, and computing the ratio
\begin{equation}
[\mathbf{Y}]_{n,m}=\left.\frac{I_n}{V_m}\right\vert_{V_k=0,\forall k\neq m},\label{eq:Ynm-def}
\end{equation}
for $n,m=1,\ldots,N$ \cite[Chapter~4]{poz11}.
In the following, we separately model the off-diagonal entries $[\mathbf{Y}]_{n,m}$, with $n\neq m$, and the diagonal entries $[\mathbf{Y}]_{m,m}$.

\subsection{Modeling the Admittance Matrix: Off-Diagonal Entries}

To derive the off-diagonal entries $[\mathbf{Y}]_{n,m}$, with $n\neq m$, given by \eqref{eq:Ynm-def}, we exploit the fact that according to the definition of $[\mathbf{Y}]_{n,m}$ in \eqref{eq:Ynm-def}, all the ports $k$ are short-circuited, with $k\neq m$.
This allows us to study a simplified circuit, without making any additional assumptions.
Specifically, we make two observations following the definition in \eqref{eq:Ynm-def}.
First, the current flowing on the transmission lines interconnecting port $j$ and port $k$, with $j,k\neq m$, is zero since these transmission lines are short-circuited at both ends according to \eqref{eq:Ynm-def}.
Consequently, these transmission lines can be removed from the circuit when computing \eqref{eq:Ynm-def}.
Second, the current flowing on the transmission lines interconnecting port $m$ and port $k$, with $k\neq n$, do not influence the current $I_n$.
Thus, also these transmission lines can be removed from the circuit as they do not impact on \eqref{eq:Ynm-def}.
Following these two observations, we have that each off-diagonal entry $[\mathbf{Y}]_{n,m}$, with $n\neq m$, can be derived by studying an equivalent circuit made of a single transmission line, as represented in Fig.~\ref{fig:model-mn}.

\begin{figure}[t]
\centering
\subfigure[]{
\includegraphics[width=0.40\textwidth]{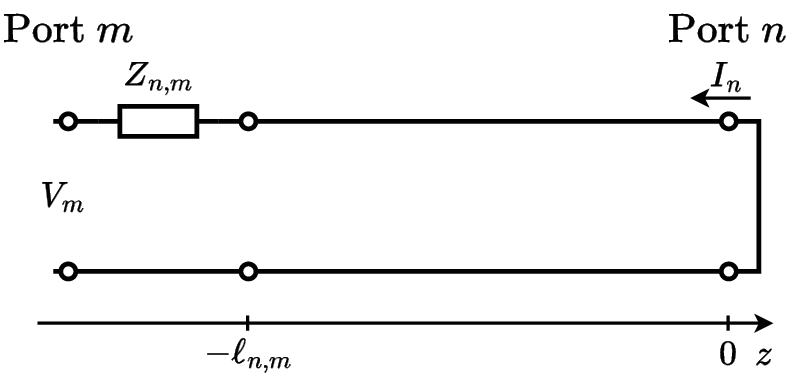}
\label{fig:model-mn}
}
\subfigure[]{
\includegraphics[width=0.44\textwidth]{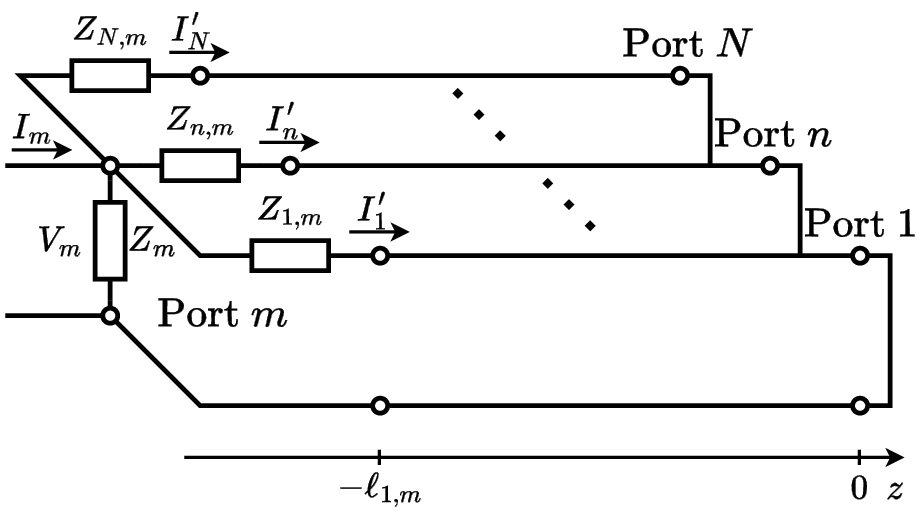}
\label{fig:model-m}
}
\caption{Equivalent circuit to be studied to compute (a) the off-diagonal entry $[\mathbf{Y}]_{n,m}$ and (b) the diagonal entry $[\mathbf{Y}]_{m,m}$. }
\label{fig:model}
\end{figure}

To study the circuit in Fig.~\ref{fig:model-mn} and determine $[\mathbf{Y}]_{n,m}$, we begin by solving the telegrapher equations in the sinusoidal steady-state condition \cite[Chapter~2]{poz11}, giving the voltage $V(z)$ and current $I(z)$ on the transmission line interconnecting port $m$ and port $n$ as
\begin{equation}
V(z)=V_0^+e^{-\gamma z}+V_0^-e^{\gamma z},\label{eq:V}
\end{equation}
\begin{equation}
I(z)=\frac{V_0^+e^{-\gamma z}-V_0^-e^{\gamma z}}{Z_0},\label{eq:I}
\end{equation}
respectively, where $V_0^+\in\mathbb{C}$ and $V_0^-\in\mathbb{C}$ are constants, $\gamma\in\mathbb{C}$ is the propagation constant of the transmission line, and $Z_0\in\mathbb{R}$ is the characteristic impedance of the transmission line\footnote{Note that $Z_0$ is real for both lossless and the so-called low-loss transmission lines \cite[Chapter~2]{poz11}, which are considered in this study.}.
In the following, we characterize the voltage and current on the transmission line by determining the constants $V_0^+$ and $V_0^-$.

Since the transmission line in Fig.~\ref{fig:model-mn} is terminated in a short circuit, the voltage at $z=0$ is zero.
Thus, by using \eqref{eq:V}, we have $V(0)=V_0^++V_0^-=0$, giving $V_0^-=-V_0^+$.
By substituting $V_0^-=-V_0^+$ into \eqref{eq:V} and \eqref{eq:I}, $V(z)$ and $I(z)$ simplify to
\begin{equation}
V(z)=V_0^+\left(e^{-\gamma z}-e^{\gamma z}\right),\label{eq:V-load}
\end{equation}
\begin{equation}
I(z)=\frac{V_0^+}{Z_0}\left(e^{-\gamma z}+e^{\gamma z}\right),\label{eq:I-load}
\end{equation}
respectively.
Furthermore, we can determine $V_0^+$ by using the input impedance of the transmission line, defined as
\begin{equation}
Z_{\text{in}}=\frac{V\left(-\ell_{n,m}\right)}{I\left(-\ell_{n,m}\right)},\label{eq:Zin1}
\end{equation}
and that can be expressed as
\begin{equation}
Z_{\text{in}}=Z_0\frac{e^{\gamma\ell_{n,m}}-e^{-\gamma\ell_{n,m}}}{e^{\gamma\ell_{n,m}}+e^{-\gamma\ell_{n,m}}},\label{eq:Zin2}
\end{equation}
by using \eqref{eq:V-load} and \eqref{eq:I-load}.
Specifically, we determine $V_0^+$ from the voltage at $z=-\ell_{n,m}$, given by
\begin{align}
V(-\ell_{n,m})
&=V_m\frac{Z_{\text{in}}}{Z_{\text{in}}+Z_{n,m}}\label{eq:Vlnm1}\\
&=V_0^+\left(e^{\gamma\ell_{n,m}}-e^{-\gamma\ell_{n,m}}\right),\label{eq:Vlnm2}
\end{align}
where \eqref{eq:Vlnm1} is a direct consequence of \eqref{eq:Zin1}, and \eqref{eq:Vlnm2} is obtained by expressing $V(-\ell_{n,m})$ through \eqref{eq:V-load}.
By equating \eqref{eq:Vlnm1} and \eqref{eq:Vlnm2}, we obtain the expression of $V_0^+$ as
\begin{align}
V_0^+
&=V_m\frac{Z_{\text{in}}}{\left(Z_{\text{in}}+Z_{n,m}\right)\left(e^{\gamma\ell_{n,m}}-e^{-\gamma\ell_{n,m}}\right)}\\
&=V_m\frac{Z_0}{Z_{n,m}\left(e^{\gamma\ell_{n,m}}+e^{-\gamma\ell_{n,m}}\right)+Z_0\left(e^{\gamma\ell_{n,m}}-e^{-\gamma\ell_{n,m}}\right)},\label{eq:V0+}
\end{align}
where \eqref{eq:V0+} is derived by using \eqref{eq:Zin2}.

We can now derive the off-diagonal entry $[\mathbf{Y}]_{n,m}$, with $n\neq m$, given by \eqref{eq:Ynm-def}.
Specifically, $I_n$ is obtained through \eqref{eq:I-load} as
\begin{equation}
I_n=-I(0)=-2\frac{V_0^+}{Z_0},\label{eq:In}
\end{equation}
with $V_0^+$ given by \eqref{eq:V0+}.
Thus, by substituting \eqref{eq:In} into \eqref{eq:Ynm-def}, we finally obtain
\begin{equation}
[\mathbf{Y}]_{n,m}=\frac{-2}{Z_{n,m}\left(e^{\gamma\ell_{n,m}}+e^{-\gamma\ell_{n,m}}\right)+Z_0\left(e^{\gamma\ell_{n,m}}-e^{-\gamma\ell_{n,m}}\right)}.\label{eq:Ynm}
\end{equation}
Remarkably, \eqref{eq:Ynm} provides the general expression of the off-diagonal entry $[\mathbf{Y}]_{n,m}$ as a function of the tunable impedance $Z_{n,m}$ and the transmission line parameters $\ell_{n,m}$, $\gamma$, and $Z_0$.
Because of the reciprocity of the circuit, $\mathbf{Y}$ is symmetric and we have $[\mathbf{Y}]_{n,m}=[\mathbf{Y}]_{m,n}$\footnote{We consider the BD-RIS circuit to be reciprocal since it is more practical to implement, as in previous works \cite{she20,ner23-1}.}.
Furthermore, if port $m$ and $n$ are not interconnected in the BD-RIS architecture, $Z_{n,m}$ is an open circuit, i.e., $Z_{n,m}=\infty$, yielding $[\mathbf{Y}]_{n,m}=[\mathbf{Y}]_{m,n}=0$.

\subsection{Modeling the Admittance Matrix: Diagonal Entries}

To derive the diagonal entries $[\mathbf{Y}]_{m,m}$, we observe that the current flowing on the transmission lines interconnecting port $j$ and port $k$, with $j,k\neq m$, is zero since these transmission lines are short-circuited at both ends according to \eqref{eq:Ynm-def}.
Thus, these transmission lines do not impact on \eqref{eq:Ynm-def} and can be removed from the circuit.
Following this observation, we have that each diagonal entry $[\mathbf{Y}]_{m,m}$ can be derived by studying the equivalent circuit represented in Fig.~\ref{fig:model-m} to find the ratio $[\mathbf{Y}]_{m,m}=I_m/V_m$.

From the circuit in Fig.~\ref{fig:model-m}, we directly notice by Kirchhoff's first law that
\begin{equation}
I_m=\frac{V_m}{Z_{m}}+\sum_{n\neq m}I_n^\prime,\label{eq:Im}
\end{equation}
where $I_n^\prime$ can be computed through \eqref{eq:I-load} as
\begin{equation}
I_n^\prime=I(-\ell_{n,m})=\frac{V_0^+}{Z_0}\left(e^{\gamma\ell_{n,m}}+e^{-\gamma\ell_{n,m}}\right),\label{eq:In-prime}
\end{equation}
$\forall n\neq m$, with $V_0^+$ given by \eqref{eq:V0+}.
Thus, by substituting \eqref{eq:Im} into \eqref{eq:Ynm-def}, we obtain
\begin{equation}
[\mathbf{Y}]_{m,m}=\frac{1}{Z_{m}}+\sum_{n\neq m}\frac{I_n^\prime}{V_m},\label{eq:Ym1}
\end{equation}
readily giving
\begin{equation}
[\mathbf{Y}]_{m,m}=\frac{1}{Z_{m}}-\sum_{n\neq m}\frac{e^{\gamma\ell_{n,m}}+e^{-\gamma\ell_{n,m}}}{2}[\mathbf{Y}]_{n,m},\label{eq:Ym}
\end{equation}
by using \eqref{eq:In-prime}, where $[\mathbf{Y}]_{n,m}$ is given by \eqref{eq:Ynm}.
Remarkably, the diagonal entry $[\mathbf{Y}]_{m,m}$ is a function of $Z_m$, $Z_{n,m}$, $\forall n\neq m$, and the transmission line parameters.

\section{Simplified Lossy BD-RIS Models}
\label{sec:sim-model}

We have modeled the RIS admittance matrix $\mathbf{Y}$ by determining its entries in \eqref{eq:Ynm} and \eqref{eq:Ym}.
In \eqref{eq:Ynm} and \eqref{eq:Ym}, the propagation constant of the transmission line $\gamma=\alpha+j\beta$ is in general a complex number, with real part $\alpha$ denoted as the attenuation constant and imaginary part $\beta$ denoted as the phase constant \cite[Chapter~2]{poz11}.
Given the complex nature of $\gamma$, it is hard to gain engineering insights into the entries of $\mathbf{Y}$ from \eqref{eq:Ynm} and \eqref{eq:Ym}.
For this reason, we simplify this general model into three simplified models by making different assumptions on the interconnection lengths $\ell_{n,m}$ and the attenuation constant $\alpha$.

\subsection{Model with Interconnection Lengths Multiple of $\lambda/2$}
\label{sec:A}

Consider the case when $\ell_{n,m}$ is a multiple of $\lambda/2$, with the wavelength $\lambda$ defined as $\lambda=\frac{2\pi}{\beta}$, i.e., when $\ell_{n,m}=\frac{\pi}{\beta}K_{n,m}$, for some $K_{n,m}\in\mathbb{Z}$\footnote{This can be achieved in practice by specifically designing the interconnections of the BD-RIS to have this desired property.
In the case of a tree-connected BD-RIS where each element is only interconnected to the adjacent ones \cite{ner23-1}, this can be obtained by considering an inter-element distance multiple of $\lambda/2$.}.
In this case, we have $e^{j\beta\ell_{n,m}}=e^{-j\beta\ell_{n,m}}=(-1)^{K_{n,m}}$.
In addition, we assume the tunable impedance components of the RIS to be purely reactive, i.e., $Z_{m}=jX_{m}$ with $X_{m}\in\mathbb{R}$ and $Z_{n,m}=jX_{n,m}$ with $X_{n,m}\in\mathbb{R}$.
Thus, \eqref{eq:Ynm} can be simplified into
\begin{align}
[\mathbf{Y}]_{n,m}
&=\frac{-2(-1)^{K_{n,m}}}{jX_{n,m}\left(e^{\alpha\ell_{n,m}}+e^{-\alpha\ell_{n,m}}\right)+Z_0\left(e^{\alpha\ell_{n,m}}-e^{-\alpha\ell_{n,m}}\right)}\label{eq:Ynm-mult1}\\
&=\frac{-(-1)^{K_{n,m}}}{jX_{n,m}\cosh\left(\alpha\ell_{n,m}\right)+Z_0\sinh\left(\alpha\ell_{n,m}\right)}.\label{eq:Ynm-mult2}
\end{align}
Remarkably, in \eqref{eq:Ynm-mult2} the term
$jX_{n,m}\cosh\left(\alpha\ell_{n,m}\right)$ is purely imaginary and the term $Z_0\sinh\left(\alpha\ell_{n,m}\right)$ is purely real.
Thus, if $[\mathbf{Y}]_{n,m}\neq0$, $[\mathbf{Y}]_{n,m}$ must have a non-zero real part, which dissipates real power.
More interestingly, it can be easily verified that the set of possible values of $[\mathbf{Y}]_{n,m}$ given by \eqref{eq:Ynm-mult2} is a circle in the complex plane with radius
\begin{equation}
r
=\frac{1}{2Z_0\sinh\left(\alpha\ell_{n,m}\right)},\label{eq:r}
\end{equation}
and centered in $c=(-(-1)^{K_{n,m}}r,0)$, where any point on the circle can be achieved by properly tuning $X_{n,m}$.
Thus, $[\mathbf{Y}]_{n,m}$ has always a negative real part when $K_{n,m}$ is even, i.e., the interconnection length $\ell_{n,m}$ is multiple of the wavelength $\lambda$, and has always a positive real part when $K_{n,m}$ is odd.
In addition, the radius $r$ decreases with the transmission line attenuation constant $\alpha$ and its length $\ell_{n,m}$, where the value of $\ell_{n,m}$ in turn depends on the RIS array geometry and the relative distance between the RIS elements $n$ and $m$.

For better clarity, we report in Fig.~\ref{fig:ynm} the possible values of $[\mathbf{Y}]_{n,m}$ considering different parameters, by fixing $Z_0=50$~$\Omega$ and $\ell_{n,m}=0.1$~m.
From Fig.~\ref{fig:ynm}, we can make the following two observations.
First, when $\alpha\ell_{n,m}=0$, we have $r=+\infty$ according to \eqref{eq:r}, meaning that $[\mathbf{Y}]_{n,m}$ can assume an arbitrary imaginary part with always zero real part.
In this ideal case, the BD-RIS has full flexibility and does not dissipate real power.
Second, $r$ decreases as the value of $\alpha\ell_{n,m}$ increases, reducing the flexibility of the BD-RIS since $[\mathbf{Y}]_{n,m}$ has imaginary part constrained within the interval $[-r,+r]$.

\begin{figure}[t]
\centering
\includegraphics[width=0.44\textwidth]{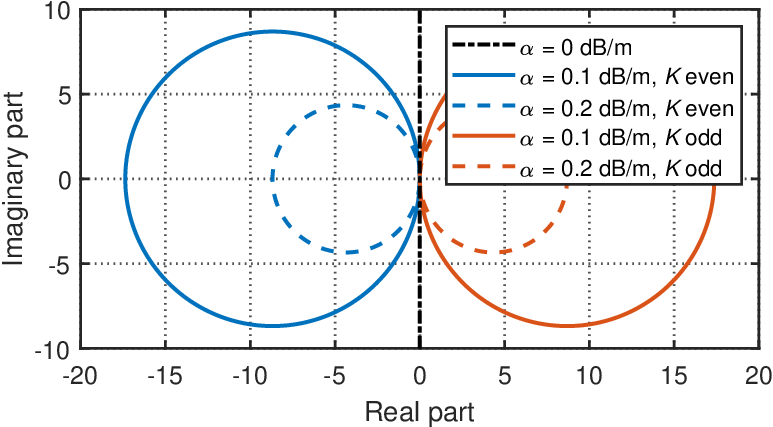}
\caption{Possible values of $[\mathbf{Y}]_{n,m}$ modeled as in \eqref{eq:Ynm-mult2}, with $Z_0=50$~$\Omega$ and $\ell_{n,m}=0.1$~m.}
\label{fig:ynm}
\end{figure}

In the case $\ell_{n,m}$ is a multiple of $\lambda/2$ for all the interconnections departing from port $m$, i.e., $\ell_{n,m}=\frac{\pi}{\beta}K_{n,m}$, for some $K_{n,m}\in\mathbb{Z}$, $\forall n\neq m$, \eqref{eq:Ym} simplifies as
\begin{align}
[\mathbf{Y}]_{m,m}
&=\frac{1}{jX_{m}}-\sum_{n\neq m}(-1)^{K_{n,m}}\cosh\left(\alpha\ell_{n,m}\right)[\mathbf{Y}]_{n,m}\label{eq:Ym-mult1}\\
&=\frac{1}{jX_{m}}+\sum_{n\neq m}\frac{1}{jX_{n,m}+Z_0\tanh\left(\alpha\ell_{n,m}\right)}.\label{eq:Ym-mult2}
\end{align}
Remarkably, $[\mathbf{Y}]_{m,m}$ has always a positive real part, $\forall K_{n,m}\in\mathbb{Z}$, given by
\begin{equation}
\Re\left\{[\mathbf{Y}]_{m,m}\right\}=\sum_{n\neq m}\frac{Z_0\tanh\left(\alpha\ell_{n,m}\right)}{X_{n,m}^2+Z_0^2\tanh^2\left(\alpha\ell_{n,m}\right)},
\end{equation}
while it can have an arbitrary imaginary part depending on the tunable value $X_{m}$.
Differently from the general model developed in Section~\ref{sec:gen-model}, this simplified model has better interpretability, and provides a good balance between tractability and physics consistency.

\subsection{Model with Lossless Interconnections}
\label{sec:B}

Consider the case of lossless interconnections, i.e., $\alpha=0$, and do not assume any constraints on the lengths $\ell_{n,m}$.
In addition, we assume purely reactive tunable impedance components, i.e., $Z_{m}=jX_{m}$ and $Z_{n,m}=jX_{n,m}$.
In this case, \eqref{eq:Ynm} can be simplified into
\begin{equation}
[\mathbf{Y}]_{n,m}=\frac{-1}{jX_{n,m}\cos\left(\beta\ell_{n,m}\right)+jZ_0\sin\left(\beta\ell_{n,m}\right)},\label{eq:Ynm-lossless}
\end{equation}
which is purely imaginary.
Interestingly, $[\mathbf{Y}]_{n,m}$ can be any imaginary number according to the tunable impedance component $X_{n,m}$.

Besides, with $\alpha=0$, \eqref{eq:Ym} boils down to
\begin{align}
[\mathbf{Y}]_{m,m}
&=\frac{1}{jX_{m}}-\sum_{n\neq m}\cos\left(\beta\ell_{n,m}\right)[\mathbf{Y}]_{n,m}\label{eq:Ym-lossless1}\\
&=\frac{1}{jX_{m}}+\sum_{n\neq m}\frac{1}{jX_{n,m}+jZ_0\tan\left(\beta\ell_{n,m}\right)},\label{eq:Ym-lossless2}
\end{align}
which is also purely imaginary.
Note that $[\mathbf{Y}]_{m,m}$ can be set independently of the values of $[\mathbf{Y}]_{n,m}$ by properly choosing $X_{m}$.
Thus, the whole admittance matrix $\mathbf{Y}$ is purely imaginary, indicating that the RIS is lossless and does not dissipate any real power.
This is consistent with the fact that the RIS is assumed to be made of lossless tunable components and lossless transmission lines.
Remarkably, when the transmission lines of the RIS interconnections are lossless, their only effect is the variation in phase of voltages and currents along them.

\subsection{Model with Interconnection Lengths Multiple of $\lambda/2$ and Lossless Interconnections}
\label{sec:C}

Consider the case when $\ell_{n,m}$ is a multiple of $\lambda/2$, i.e., when $\ell_{n,m}=\frac{\pi}{\beta}K_{n,m}$, for some $K_{n,m}\in\mathbb{Z}$, and the interconnections are lossless, i.e., $\alpha=0$.
In this case, we have $e^{\gamma\ell_{n,m}}=e^{-\gamma\ell_{n,m}}=(-1)^{K_{n,m}}$.
In addition, we assume purely reactive tunable impedance components, i.e., $Z_{m}=jX_{m}$ and $Z_{n,m}=jX_{n,m}$.
Thus, \eqref{eq:Ynm} can be simplified into
\begin{equation}
[\mathbf{Y}]_{n,m}=\frac{-(-1)^{K_{n,m}}}{jX_{n,m}},\label{eq:Ynm-ideal}
\end{equation}
which can be an arbitrary imaginary number depending on the tunable impedance component $X_{n,m}$.

In the case $\ell_{n,m}$ is a multiple of $\lambda/2$ for all the interconnections departing from port $m$, i.e., $\ell_{n,m}=\frac{\pi}{\beta}K_{n,m}$, for some $K_{n,m}\in\mathbb{Z}$, $\forall n\neq m$, \eqref{eq:Ym} simplifies as
\begin{align}
[\mathbf{Y}]_{m,m}
&=\frac{1}{jX_{m}}-\sum_{n\neq m}(-1)^{K_{n,m}}[\mathbf{Y}]_{n,m}\label{eq:Ym-ideal1}\\
&=\frac{1}{jX_{m}}+\sum_{n\neq m}\frac{1}{jX_{n,m}},\label{eq:Ym-ideal2}
\end{align}
which can be an arbitrary imaginary number.
Remarkably, when all the interconnection lengths $\ell_{n,m}$ are a multiple of $\lambda$, i.e., $K_{n,m}$ is even, $\forall n,m$, the transmission lines have no visible effect as they apply a phase shift of a multiple of $2\pi$ to the voltages and currents propagating along them.
In this case, the simplified model in \eqref{eq:Ynm-ideal} and \eqref{eq:Ym-ideal2} is equivalent to the model with no transmission lines, i.e., with $\ell_{n,m}=0$, widely considered in previous literature on BD-RIS \cite{she20,ner23-1}.

\section{Lossy BD-RIS Optimization}
\label{sec:optimization}

We have modeled the admittance matrix $\mathbf{Y}$ of BD-RIS with lossy interconnections and derived three simplified models.
In this section, we optimize BD-RIS characterized by these simplified models to assess the impact of losses on the performance of BD-RIS.

\subsection{Optimizing Lossless BD-RIS}

Consider the lossless BD-RIS model of Section~\ref{sec:B}.
The received signal power maximization problem in a BD-RIS-aided \gls{mimo} system writes as
\begin{align}
\underset{X_{n,m},X_m,\mathbf{g},\mathbf{w}}{\mathsf{\mathrm{max}}}\;\;
&P_{T}\left\vert\mathbf{g}\left(\mathbf{H}_{RT}+\mathbf{H}_{R}\boldsymbol{\Theta}\mathbf{H}_{T}\right)\mathbf{w}\right\vert^{2}\label{eq:p1-lossless-obj}\\
\mathsf{\mathrm{s.t.}}\;\;\;
&\eqref{eq:T(Y)},\;\eqref{eq:Ynm-lossless},\;\eqref{eq:Ym-lossless2},\label{eq:p1-lossless-c1}\\
&\left\Vert\mathbf{g}\right\Vert=1,\;\left\Vert\mathbf{w}\right\Vert=1,\label{eq:p1-lossless-c2}
\end{align}
which is solved by jointly optimizing $X_{n,m}$, $X_{m}$, $\mathbf{g}$, $\mathbf{w}$.
To this end, we initialize $\mathbf{g}$ and $\mathbf{w}$ to feasible values and alternate between the following two steps.

\subsubsection{Optimizing $X_{n,m}$ and $X_{m}$ with Fixed $\mathbf{g}$ and $\mathbf{w}$}

When $\mathbf{g}$ and $\mathbf{w}$ are fixed, we update $X_{n,m}$ and $X_{m}$ by solving
\begin{align}
\underset{X_{n,m},X_m}{\mathsf{\mathrm{max}}}\;\;
&P_T\left\vert \bar{h}_{RT}+\bar{\mathbf{h}}_{R}\boldsymbol{\Theta}\bar{\mathbf{h}}_{T}\right\vert^2\label{eq:p2-lossless-obj}\\
\mathsf{\mathrm{s.t.}}\;\;\;
&\eqref{eq:T(Y)},\;\eqref{eq:Ynm-lossless},\;\eqref{eq:Ym-lossless2},\label{eq:p2-lossless-c}
\end{align}
where $\bar{h}_{RT}=\mathbf{g}\mathbf{H}_{RT}\mathbf{w}$, $\bar{\mathbf{h}}_{R}=\mathbf{g}\mathbf{H}_{R}$, and $\bar{\mathbf{h}}_{T}=\mathbf{H}_{T}\mathbf{w}$, which is apparently hard given the non-convex constraints.
Nevertheless, \eqref{eq:p2-lossless-obj}-\eqref{eq:p2-lossless-c} can be simplified by noticing that $\mathbf{Y}$ given by \eqref{eq:Ynm-lossless} and \eqref{eq:Ym-lossless2} can be any purely imaginary symmetric matrix with $[\mathbf{Y}]_{n,m}=0$ if ports $m$ and $n$ are not interconnected, i.e., $\mathbf{Y}=j\mathbf{B}$, $\mathbf{B}=\mathbf{B}^T$, with $[\mathbf{B}]_{n,m}=0$ if $Z_{n,m}=\infty$, where $\mathbf{B}\in\mathbb{R}^{N\times N}$ denotes the RIS susceptance matrix.
Thus, \eqref{eq:p2-lossless-obj}-\eqref{eq:p2-lossless-c} can be rewritten as
\begin{align}
\underset{\mathbf{B}}{\mathsf{\mathrm{max}}}\;\;
&P_{T}\left\vert\bar{h}_{RT}+\bar{\mathbf{h}}_{R}\left(\mathbf{I}+jZ_{0}\mathbf{B}\right)^{-1}\left(\mathbf{I}-jZ_{0}\mathbf{B}\right)\bar{\mathbf{h}}_{T}\right\vert^{2}\label{eq:p3-lossless-obj}\\
\mathsf{\mathrm{s.t.}}\;\;\;
&\mathbf{B}=\mathbf{B}^T,\;[\mathbf{B}]_{n,m}=0\text{ if }Z_{n,m}=\infty,\label{eq:p3-lossless-c}
\end{align}
which is an unconstrained problem in the upper triangular entries of $\mathbf{B}$ that are not forced to zero.
This problem has been solved through the quasi-Newton method in \cite{she20} for group-/fully-connected RISs, and through global optimal closed-form solutions in \cite{ner22} and \cite{ner23-1}, for group-/fully- and forest-/tree-connected RISs, respectively.
Given the optimal susceptance matrix $\mathbf{B}^\star$ solution to \eqref{eq:p3-lossless-obj}-\eqref{eq:p3-lossless-c}, the values of the tunable reactance components $X_{n,m}$ and $X_m$ can be obtained by inverting \eqref{eq:Ynm-lossless} and \eqref{eq:Ym-lossless1} with $\mathbf{Y}=j\mathbf{B}^\star$, respectively, i.e.,
\begin{equation}
X_{n,m}=\frac{1}{\cos\left(\beta\ell_{n,m}\right)[\mathbf{B}^\star]_{n,m}}-Z_0\tan\left(\beta\ell_{n,m}\right),
\end{equation}
\begin{equation}
X_{m}=\frac{-1}{[\mathbf{B}^\star]_{m,m}+\sum_{n\neq m}\cos\left(\beta\ell_{n,m}\right)[\mathbf{B}^\star]_{n,m}}.
\end{equation}
Thus, the phase change of the signals along the transmission lines, accounted by the term $\beta\ell_{n,m}$, does not change the achievable received signal power by \eqref{eq:p3-lossless-obj}-\eqref{eq:p3-lossless-c}.

\subsubsection{Optimizing $\mathbf{g}$ and $\mathbf{w}$ with Fixed $X_{n,m}$ and $X_{m}$}

When $X_{n,m}$ and $X_{m}$ are fixed, also $\boldsymbol{\Theta}$ is fixed, as given by \eqref{eq:T(Y)}, where $\mathbf{Y}$ is determined by \eqref{eq:Ynm-lossless} and \eqref{eq:Ym-lossless2}.
Thus, $\mathbf{g}$ and $\mathbf{w}$ are updated as the dominant left and right singular vectors of $\mathbf{H}_{RT}+\mathbf{H}_{R}\boldsymbol{\Theta}\mathbf{H}_{T}$, respectively, which is globally optimal.

These two steps are alternatively repeated until convergence of the objective \eqref{eq:p1-lossless-obj}.
Notably, this optimization framework is also valid for the model of Section~\ref{sec:C} as it is a special case of the model of Section~\ref{sec:B}.
The complexity of step 1) is given by the complexity of the solutions proposed in \cite{ner22} and \cite{ner23-1}, which is $\mathcal{O}(N^3)$ for both fully- and tree-connected RISs.
Besides, the complexity of step 2) is given by the complexity of computing the singular value decomposition of $\mathbf{H}_{RT}+\mathbf{H}_{R}\boldsymbol{\Theta}\mathbf{H}_{T}$, which is $\mathcal{O}(N_T^3)$, in the case $N_T=N_R$.
Since a RIS is expected to have many more elements than the number of transit antennas, the complexity of our algorithm is dominated by $\mathcal{O}(N^3)$.

\subsection{Optimizing Lossy BD-RIS}

Considering the lossy BD-RIS model of Section~\ref{sec:A}, the received signal power maximization problem writes as
\begin{align}
\underset{X_{n,m},X_m,\mathbf{g},\mathbf{w}}{\mathsf{\mathrm{max}}}\;\;
&P_{T}\left\vert\mathbf{g}\left(\mathbf{H}_{RT}+\mathbf{H}_{R}\boldsymbol{\Theta}\mathbf{H}_{T}\right)\mathbf{w}\right\vert^{2}\label{eq:p1-lossy-obj}\\
\mathsf{\mathrm{s.t.}}\;\;\;
&\eqref{eq:T(Y)},\;\eqref{eq:Ynm-mult2},\;\eqref{eq:Ym-mult2},\label{eq:p1-lossy-c1}\\
&\left\Vert\mathbf{g}\right\Vert=1,\;\left\Vert\mathbf{w}\right\Vert=1,\label{eq:p2-lossy-c2}
\end{align}
which is solved by initializing $\mathbf{g}$ and $\mathbf{w}$ to feasible values and by alternating between the following two steps.

\subsubsection{Optimizing $X_{n,m}$ and $X_{m}$ with Fixed $\mathbf{g}$ and $\mathbf{w}$}

When $\mathbf{g}$ and $\mathbf{w}$ are fixed, $X_{n,m}$ and $X_{m}$ are updated by solving
\begin{equation}
\underset{X_{n,m},X_m}{\mathsf{\mathrm{max}}}\;\;
P_T\left\vert\bar{h}_{RT}+\bar{\mathbf{h}}_{R}\boldsymbol{\Theta}\bar{\mathbf{h}}_{T}\right\vert^2\;\;
\mathsf{\mathrm{s.t.}}\;\;
\eqref{eq:T(Y)},\;\eqref{eq:Ynm-mult2},\;\eqref{eq:Ym-mult2}.\label{eq:p2-lossy}
\end{equation}
Interestingly, the constraints of problem \eqref{eq:p2-lossy} can be tightened as a consequence of the following proposition.

\begin{proposition}
The average real power dissipated by an RIS with admittance matrix $\mathbf{Y}$ modeled as in Section~\ref{sec:A} is reduced if any $\vert\Re\{[\mathbf{Y}]_{i,j}\}\vert$ is reduced, with $i\neq j$.\label{pro}
\end{proposition}
\begin{proof}
Consider the following two RISs: the first has admittance matrix $\mathbf{Y}$ and dissipates a power $P_{\text{av}}$; the second has admittance matrix $\mathbf{Y}^\prime$ obtained from $\mathbf{Y}$ by reducing $\vert\Re\{[\mathbf{Y}]_{i,j}\}\vert$ by $\Delta>0$ for some $i,j$, i.e., $\vert\Re\{[\mathbf{Y}^\prime]_{i,j}\}\vert=\vert\Re\{[\mathbf{Y}]_{i,j}\}\vert-\Delta$, and dissipates a power $P_{\text{av}}^\prime$.
According to \eqref{eq:Ynm-mult2} and \eqref{eq:Ym-mult1}, introducing $\mathbf{G}=\Re\{\mathbf{Y}\}$ and $\mathbf{G}^\prime=\Re\{\mathbf{Y}^\prime\}$, we have $[\mathbf{G}^\prime]_{n,m}=[\mathbf{G}]_{n,m}$ for $m,n\notin\{i,j\}$, and
\begin{equation}
\left[\mathbf{G}^\prime\right]_{i,j}=\left[\mathbf{G}\right]_{i,j}+(-1)^{K_{i,j}}\Delta,\label{eq:Gij}
\end{equation}
\begin{equation}
\left[\mathbf{G}^\prime\right]_{j,i}=\left[\mathbf{G}\right]_{j,i}+(-1)^{K_{i,j}}\Delta,\label{eq:Gji}
\end{equation}
\begin{equation}
\left[\mathbf{G}^\prime\right]_{k,k}=\left[\mathbf{G}\right]_{k,k}-\cosh\left(\alpha\ell_{i,j}\right)\Delta,\label{eq:Gkk}
\end{equation}
for $k\in\{i,j\}$.
Given these two RISs, our goal is to prove that $P_{\text{av}}^\prime\leq P_{\text{av}}$.
To this end, we first recall that average real power dissipated by an $N$-port network is defined as $P_{\text{av}}=\Re\{\mathbf{v}^T\mathbf{i}^*\}/2$, where $\mathbf{v},\mathbf{i}\in\mathbb{C}^{N\times1}$ are the voltage and current vectors at the $N$ ports, respectively \cite[Chapter 4]{poz11}.
Thus, if the admittance matrix $\mathbf{Y}$ of this $N$-port network is symmetric, it is possible to show that
\begin{equation}
P_{\text{av}}
=\frac{1}{2}\mathbf{v}^T\mathbf{G}\mathbf{v}^*
=\frac{1}{2}\sum_{m,n}v_nv_m^*\left[\mathbf{G}\right]_{n,m},\label{eq:Pav}
\end{equation}
where $\mathbf{v}=[v_1,\ldots,v_N]^T$, giving the power dissipated by the first RIS.
Similarly, the power dissipated by the second RIS is
\begin{align}
P_{\text{av}}^\prime
&=\frac{1}{2}\sum_{m,n}v_nv_m^*\left[\mathbf{G}^\prime\right]_{n,m}\\
&=
\frac{1}{2}\sum_{m,n\notin\{i,j\}}v_nv_m^*\left[\mathbf{G}       \right]_{n,m}+
\frac{1}{2}\sum_{m,n\in   \{i,j\}}v_nv_m^*\left[\mathbf{G}^\prime\right]_{n,m},\label{eq:Pav-prime}
\end{align}
since $\left[\mathbf{G}^\prime\right]_{n,m}=\left[\mathbf{G}\right]_{n,m}$ for $m,n\notin\{i,j\}$.
By substituting \eqref{eq:Gij}, \eqref{eq:Gji}, and \eqref{eq:Gkk} into \eqref{eq:Pav-prime}, we obtain
\begin{multline}
P_{\text{av}}^\prime=P_{\text{av}}-\frac{1}{2}\left(
      \vert v_i\vert^2\cosh\left(\alpha\ell_{i,j}\right)\Delta+\right.\\
\left.\vert v_j\vert^2\cosh\left(\alpha\ell_{i,j}\right)\Delta-
(-1)^{K_{i,j}}\left(v_iv_j^*+v_i^*v_j\right)\Delta\right).
\end{multline}
Finally, since $\cosh\left(\alpha\ell_{i,j}\right)\geq1$, we have
\begin{equation}
P_{\text{av}}^\prime\leq P_{\text{av}}-\frac{1}{2}\vert v_i-(-1)^{K_{i,j}}v_j\vert^2\Delta,
\end{equation}
indicating that $P_{\text{av}}^\prime\leq P_{\text{av}}$.
\end{proof}

From Proposition~\ref{pro}, we deduce that it is convenient to select the off-diagonal entries $[\mathbf{Y}]_{n,m}$ with a small real part in magnitude to reduce the power dissipated by the RIS.
Interestingly, Fig.~\ref{fig:ynm} shows that we can select $[\mathbf{Y}]_{n,m}$ with any imaginary part in the interval $[-r,+r]$, with the real part constrained by $\vert\Re\{[\mathbf{Y}]_{n,m}\}\vert\leq r$.
This constraint means that we can select $[\mathbf{Y}]_{n,m}$ in the right semicircle of its possible values if $K_{n,m}$ is even, and in the left semicircle if $K_{n,m}$ is odd, to reduce the dissipated power.

By adding the constraint $\vert\Re\{[\mathbf{Y}]_{n,m}\}\vert\leq r$ for $n\neq m$ to \eqref{eq:p2-lossy}, the resulting problem can be reformulated by expressing $\mathbf{Y}$ as a function of an auxiliary variable $\mathbf{A}\in\mathbb{R}^{N\times N}$ as
\begin{equation}
[\mathbf{Y}]_{n,m}=(-1)^{K_{n,m}}\left(\frac{r}{\sqrt{\frac{[\mathbf{A}]_{n,m}^2}{r^2}+1}}-r\right)+j\frac{[\mathbf{A}]_{n,m}}{\sqrt{\frac{[\mathbf{A}]_{n,m}^2}{r^2}+1}},\label{eq:Anm}
\end{equation}
\begin{equation}
[\mathbf{Y}]_{m,m}=\sum_{n\neq m}\cosh\left(\alpha\ell_{n,m}\right)\left\vert\Re\{[\mathbf{Y}]_{n,m}\}\right\vert+j[\mathbf{A}]_{m,m},\label{eq:Am}
\end{equation}
where $\mathbf{A}$ has been introduced such that $\Im\{[\mathbf{Y}]_{n,m}\}$ varies in the interval $[-r,+r]$ as $[\mathbf{A}]_{n,m}$ varies in $[-\infty,+\infty]$, with $n\neq m$, and $\Im\{[\mathbf{Y}]_{m,m}\}=[\mathbf{A}]_{m,m}$.
Furthermore, in the case of a lossless RIS, i.e., $r=+\infty$, \eqref{eq:Anm} and \eqref{eq:Am} boil down to $[\mathbf{Y}]_{n,m}=j[\mathbf{A}]_{n,m}$ and $[\mathbf{Y}]_{m,m}=j[\mathbf{A}]_{m,m}$, respectively.
Thus, the resulting problem writes as
\begin{align}
\underset{\mathbf{A}}{\mathsf{\mathrm{max}}}\;
&P_{T}\left\vert\bar{h}_{RT}+\bar{\mathbf{h}}_{R}\left(\mathbf{I}+Z_{0}\mathbf{Y}\right)^{-1}\left(\mathbf{I}-Z_{0}\mathbf{Y}\right)\bar{\mathbf{h}}_{T}\right\vert^{2}\label{eq:p3-lossy-obj}\\
\mathsf{\mathrm{s.t.}}\;\;
&\eqref{eq:Anm},\;\eqref{eq:Am},\;\mathbf{A}=\mathbf{A}^T,\;[\mathbf{A}]_{n,m}=0\text{ if }Z_{n,m}=\infty,\label{eq:p3-lossy-c}
\end{align}
which is an unconstrained problem in the upper triangular entries of $\mathbf{A}$ that are not forced to zero.
Remarkably, \eqref{eq:p3-lossy-obj}-\eqref{eq:p3-lossy-c} is a generalization of \eqref{eq:p3-lossless-obj}-\eqref{eq:p3-lossless-c} since in the case of a lossless RIS, constraints \eqref{eq:Anm} and \eqref{eq:Am} become $[\mathbf{Y}]_{n,m}=j[\mathbf{A}]_{n,m}$ and $[\mathbf{Y}]_{m,m}=j[\mathbf{A}]_{m,m}$, respectively, reducing \eqref{eq:p3-lossy-obj}-\eqref{eq:p3-lossy-c} to \eqref{eq:p3-lossless-obj}-\eqref{eq:p3-lossless-c}.
Since \eqref{eq:p3-lossy-obj}-\eqref{eq:p3-lossy-c} is unconstrained, it can be solved with the quasi-Newton method to find a stationary point.
Given the optimal $\mathbf{A}^\star$ solution to \eqref{eq:p3-lossy-obj}-\eqref{eq:p3-lossy-c}, the admittance matrix $\mathbf{Y}$ is first obtained throguh \eqref{eq:Anm} and \eqref{eq:Am}.
Then, the values of the tunable reactance components $X_{n,m}$ and $X_m$ are obtained from $\mathbf{Y}$ by inverting \eqref{eq:Ynm-mult2} and \eqref{eq:Ym-mult1}, respectively.

\subsubsection{Optimizing $\mathbf{g}$ and $\mathbf{w}$ with Fixed $X_{n,m}$ and $X_{m}$}

When $X_{n,m}$ and $X_{m}$ are fixed, also $\boldsymbol{\Theta}$ is fixed, as given by \eqref{eq:T(Y)}, where $\mathbf{Y}$ is determined by \eqref{eq:Ynm-mult2} and \eqref{eq:Ym-mult2}.
Thus, $\mathbf{g}$ and $\mathbf{w}$ are updated as the dominant left and right singular vectors of $\mathbf{H}_{RT}+\mathbf{H}_{R}\boldsymbol{\Theta}\mathbf{H}_{T}$, respectively, which is globally optimal.

These two steps are repeated until convergence of \eqref{eq:p1-lossy-obj}.
The complexity of step 1) is given by the complexity of solving the unconstrained problem in \eqref{eq:p3-lossy-obj}-\eqref{eq:p3-lossy-c} via the quasi-Newton method.
Considering the quasi-Newton method with BFGS update, the computational
complexity for each iteration grows with the square of the number of variables optimized \cite{she20}.
For fully-connected RISs, where we have $N(N+1)/2$ tunable components, the complexity is $\mathcal{O}(N^2(N+1)^2/4)$, and for tree-connected RISs, where we have $2N-1$ tunable components, the complexity is $\mathcal{O}((2N-1)^2)$.
Besides, the complexity of step 2) is given by $\mathcal{O}(N_T^3)$, assuming $N_T=N_R$.
Since $N\gg N_T$ in practice, our algorithm has complexity $\mathcal{O}(N^4)$ for fully-connected RISs and $\mathcal{O}(N^2)$ for tree-connected RISs.

\section{Numerical Results}
\label{sec:results}

\begin{figure*}[t]
\centering
\subfigure[]{
\includegraphics[height=0.32\textwidth]{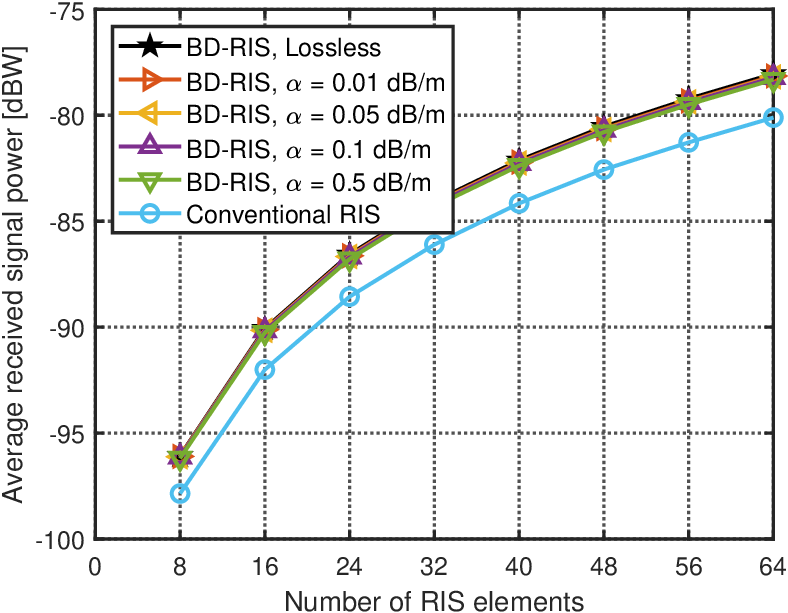}
\label{fig:loc-N-pow}
}
\subfigure[]{
\includegraphics[height=0.32\textwidth]{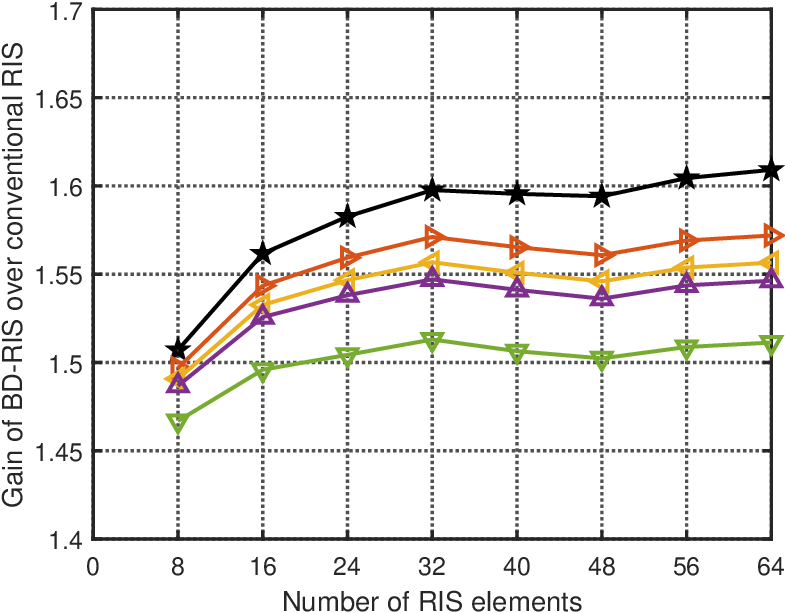}
\label{fig:loc-N-gain}
}
\caption{(a) Received signal power and (b) gain of BD-RIS over conventional RIS versus the number of RIS elements, for localized RIS.}
\label{fig:loc-N}
\end{figure*}

\begin{figure*}[t]
\centering
\subfigure[]{
\includegraphics[height=0.32\textwidth]{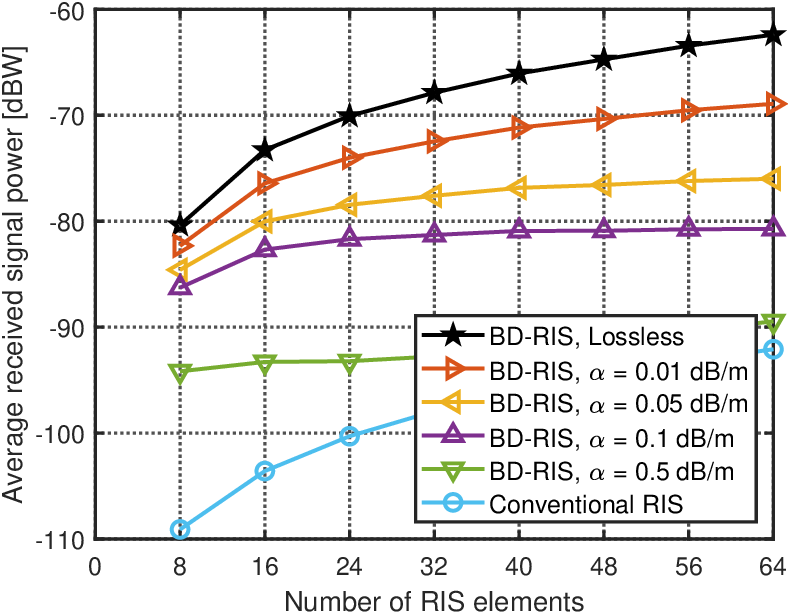}
\label{fig:dis-N-pow}
}
\subfigure[]{
\includegraphics[height=0.32\textwidth]{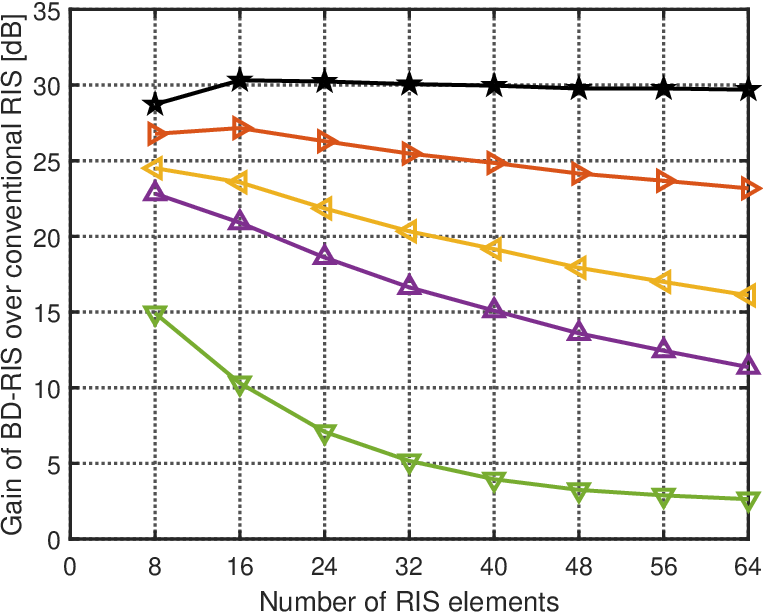}
\label{fig:dis-N-gain}
}
\caption{(a) Received signal power and (b) gain of BD-RIS over conventional RIS versus the number of RIS elements, for distributed RIS.}
\label{fig:dis-N}
\end{figure*}

\begin{figure*}[t]
\centering
\subfigure[]{
\includegraphics[height=0.32\textwidth]{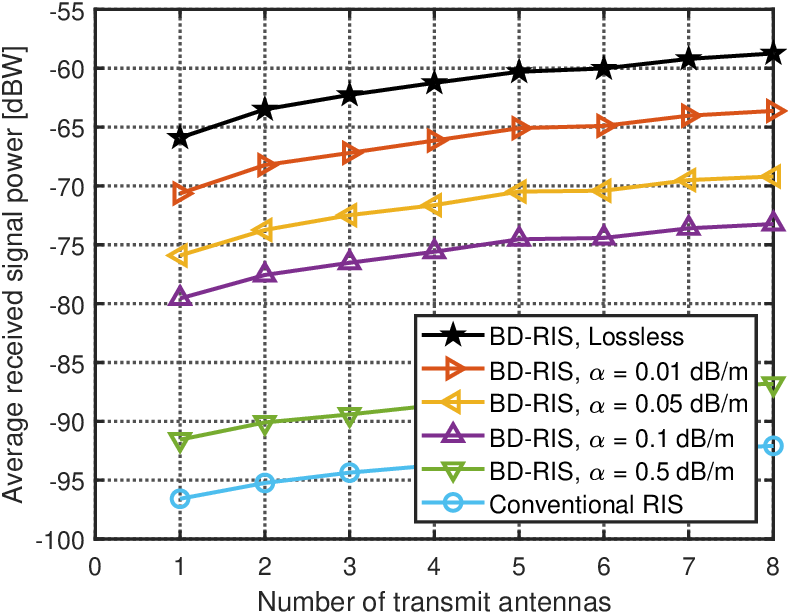}
\label{fig:dis-NT-pow}
}
\subfigure[]{
\includegraphics[height=0.32\textwidth]{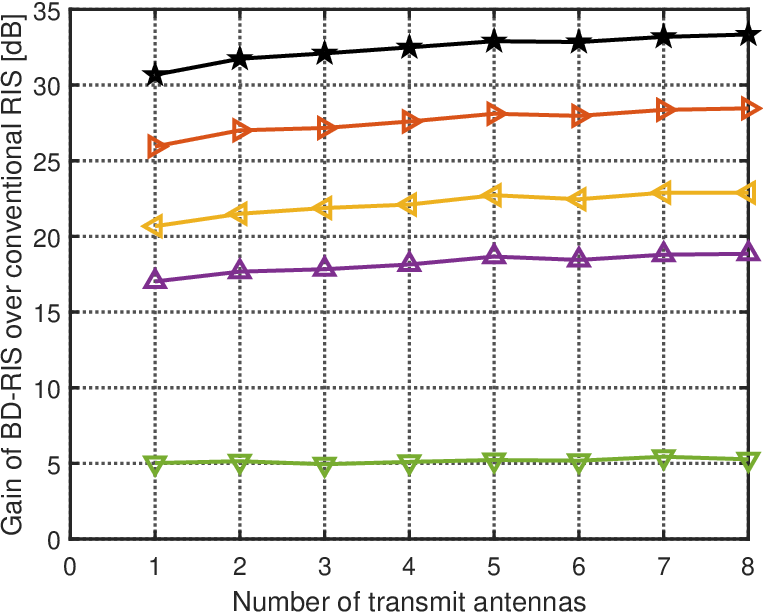}
\label{fig:dis-NT-gain}
}
\caption{(a) Received signal power and (b) gain of BD-RIS over conventional RIS versus the number of transmit antennas, for distributed RIS.}
\label{fig:dis-NT}
\end{figure*}

\begin{figure*}[t]
\centering
\subfigure[]{
\includegraphics[height=0.32\textwidth]{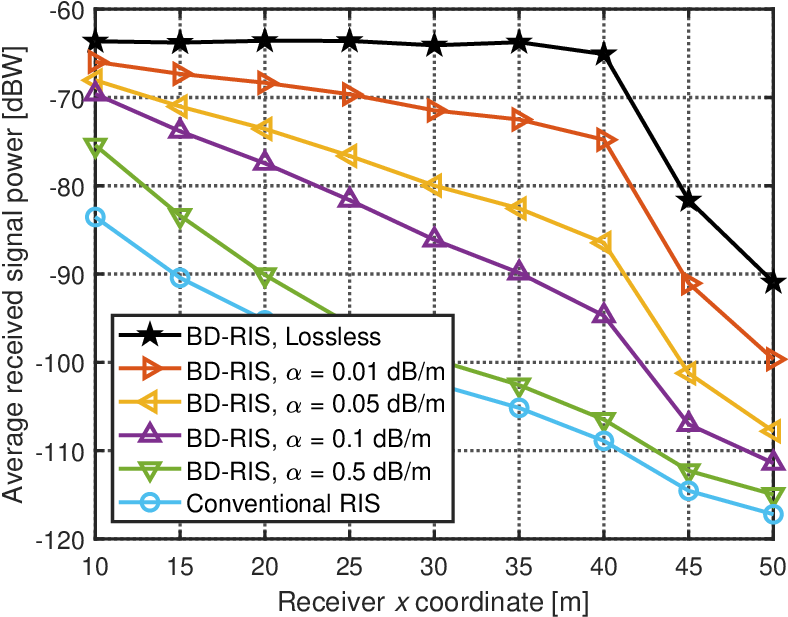}
\label{fig:dis-x-pow}
}
\subfigure[]{
\includegraphics[height=0.32\textwidth]{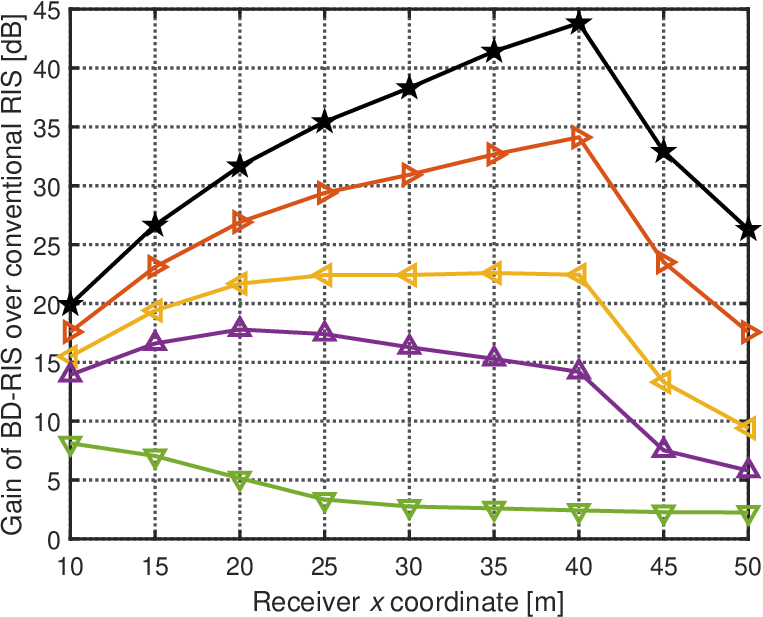}
\label{fig:dis-x-gain}
}
\caption{(a) Received signal power and (b) gain of BD-RIS over conventional RIS versus the $x$ coordinate of the receiver, for distributed RIS.}
\label{fig:dis-x}
\end{figure*}

To evaluate the performance of localized and distributed BD-RIS with lossy interconnections, we consider the system represented in Fig.~\ref{fig:deployment}, where the transmitter and receiver are located at $(0,0,0)$ and $(20,0,0)$, respectively.
In the case of localized RIS, the RIS is a \gls{ula} centered in $(20,0,2)$ and with an inter-element distance $0.05$~m, as shown in Fig.~\ref{fig:deployment-loc}.
We model the path-gain of the channels $\mathbf{H}_{R}^{\text{Loc}}$ and $\mathbf{H}_{T}^{\text{Loc}}$ through the distance-dependent model and their small-scale fading as \gls{iid} Rayleigh distributed, i.e., $\mathbf{H}_{i}^{\text{Loc}}=\sqrt{\rho_{i}}\widetilde{\mathbf{H}}_{i}$, where $\rho_{i}=C_{0}d_{i}^{-a}$ and $\widetilde{\mathbf{H}}_{i}\sim\mathcal{CN}(\mathbf{0},\mathbf{I})$, for $i\in\{R,T\}$.
In the case of distributed RIS, the RIS is a \gls{ula} with elements uniformly placed between $(0,0,2)$ and $(40,0,2)$, as shown in Fig.~\ref{fig:deployment-dis}.
Accordingly, the channels are given by $\mathbf{H}_{R}^{\text{Dis}}=\widetilde{\mathbf{H}}_{R}\mathbf{R}_{R}^{1/2}$ and $\mathbf{H}_{T}^{\text{Dis}}=\mathbf{R}_{T}^{1/2}\widetilde{\mathbf{H}}_{T}$, where $\mathbf{R}_{i}=\text{diag}(\boldsymbol{\rho}_{i})$, $[\boldsymbol{\rho}_{i}]_{n}=C_{0}[\mathbf{d}_{i}]_{n}^{-a}$, and $\widetilde{\mathbf{H}}_{i}\sim\mathcal{CN}(\mathbf{0},\mathbf{I})$, for $i\in\{R,T\}$.
The direct channel is assumed to be fully obstructed, i.e., $\mathbf{H}_{RT}=\mathbf{0}$, and we set $C_{0}=-30$~dB, $a=4$, and $P_{T}=10$~W.
In the following, we report the performance of tridiagonal BD-RIS, a tree-connected BD-RIS that is proven to be the least complex BD-RIS architecture able to achieve the performance upper bound in single-user systems \cite{ner23-1}, and compare it with conventional RIS.
Recalling that tridiagonal BD-RIS presents interconnections only between adjacent RIS elements \cite{ner23-1}, its interconnection lengths are all equal to the inter-element distance.
Furthermore, we recall that the circuit complexity of a tridiagonal BD-RIS is $2N-1$ tunable impedance components \cite{ner23-1}, while the circuit complexity of a single-connected RIS is given by $N$ tunable impedance components (each connecting a RIS element to ground).

In Fig.~\ref{fig:loc-N}, we report the received signal power obtained with localized RIS versus the number of RIS elements, with $N_R=1$ and $N_T=1$ for simplicity.
We can make the following observations.
\textit{First}, BD-RIS achieves a higher received signal power than conventional RIS thanks to the additional flexibility, both with lossless and lossy interconnections.
\textit{Second}, lossless BD-RIS achieves the highest received signal power since no power is dissipated by the BD-RIS circuit.
Specifically, Fig.~\ref{fig:loc-N-gain} shows that lossless BD-RIS offers a gain over conventional RIS of up to 1.62, as discussed in Section~\ref{sec:G-Loc}.
\textit{Third}, the received signal power achieved by BD-RIS decreases as the attenuation constant $\alpha$ increases.
However, the performance of localized BD-RIS is only slightly impacted by the losses given the short inter-element distance.

In Fig.~\ref{fig:dis-N}, we report the received signal power obtained with distributed RIS versus the number of RIS elements, with $N_R=1$ and $N_T=1$.
We can make the following observations.
\textit{First}, BD-RIS always achieves a higher received signal power than conventional RIS.
\textit{Second}, lossless BD-RIS achieves the highest received signal power, offering gains of up to 30~dB over conventional RIS.
These massive gains are because the \gls{em} signal can propagate without losses between the BD-RIS elements, which are spread over a long line.
\textit{Third}, the received signal power of distributed BD-RIS visibly decreases as the attenuation constant $\alpha$ increases, given the significantly long inter-element distance.
Nevertheless, lossy BD-RIS can still achieve more than an order of magnitude of gain over conventional RIS.
\textit{Fourth}, comparing Fig.~\ref{fig:loc-N}(a) and Fig.~\ref{fig:dis-N}(a), we notice that increasing the number of elements in a localized conventional RIS can potentially match the performance of a distributed BD-RIS.
For example, a received signal power of $P_R=-80$~dBW can be obtained by a localized conventional RIS with $N=64$ RIS elements ($64$ tunable impedance components) or by a distributed BD-RIS with $\alpha=0.05$~dB/m with $N=16$ RIS elements ($2N-1=31$ tunable impedance components).
However, achieving this may demand significantly more tunable components and result in greater control overhead compared to a distributed BD-RIS.

To investigate the impact of $N_R$ and $N_T$ on the system performance, we report in Fig.~\ref{fig:dis-NT} the received signal power of distributed RIS versus the number of transmit antennas, with $N=32$ and $N_R=2$.
We observe that the received signal power of BD-RIS and conventional RIS increases with $N_T$, given the diversity gain offered by the multiple transmit antennas.
Besides, the gain of BD-RIS over conventional RIS is approximately stable for any value of $N_T$.

In Fig.~\ref{fig:dis-x}, we report the received signal power obtained with distributed RIS for different locations of the receiver.
Specifically, the receiver is located in $(x,0,0)$, where $x\in[10,50]$, while we fix $N=32$, $N_R=2$, and $N_T=2$.
We observe that the received signal power decreases as $x$ increases since $x$ also represents the distance between the transmitter and receiver.
The received signal power decreases significantly after $x=40$~m, as this is the location of the last distributed RIS element.
Besides, Fig.~\ref{fig:dis-x-gain} shows that the gain of lossless BD-RIS over conventional RIS increases with $x$ reaching values as high as several orders of magnitude until $x=40$~m.

\section{Conclusion}
\label{sec:conclusion}

We propose the concept of distributed RIS, whose elements are distributed over a wide region, in opposition to localized RIS, as commonly considered in previous literature.
We derive its scaling laws and analyze the gain of distributed RIS over localized RIS, and of distributed BD-RIS over distributed conventional RIS.
The derived gains show that lossless distributed BD-RIS can offer gains as high as several orders of magnitude over distributed conventional RIS and localized BD-RIS.
In particular, these substantial gains are enabled by the tunable impedance components interconnecting the BD-RIS elements, which allow the \gls{em} signal to effectively reach the receiver by propagating within the BD-RIS circuit.

To assess the practical performance of distributed BD-RIS, we model BD-RIS with lossy interconnections by using transmission line theory.
More precisely, we model the BD-RIS interconnections as tunable impedance components in series with lossy transmission lines and derive the expression of the BD-RIS admittance matrix.
Since it is hard to gain engineering insights from the obtained model, we also derive three simplified models through different assumptions.
Finally, we optimize lossy BD-RIS based on the proposed models and evaluate its performance.
Numerical results show that the performance of localized BD-RIS is only slightly impacted by losses because of the short interconnection lengths.
Furthermore, distributed BD-RIS can achieve orders of magnitude of gains over conventional RIS, even with low losses.
We identify two future research directions.
First, the impact of the delay spread due to long interconnections on the performance should be investigated by considering a frequency-selective channel model.
Second, the optimal deployment of both localized and distributed RISs should be explored to maximize the coverage in a given environment.

\bibliographystyle{IEEEtran}
\bibliography{IEEEabrv,main}

\end{document}